\documentclass[a4paper,10pt]{article}

\usepackage{amssymb}
\usepackage{amsmath}
\usepackage{subfigure}
\usepackage{slashed} 
\usepackage{tensor}

\usepackage{physics}
\allowdisplaybreaks

\usepackage{soul}

\usepackage{graphicx}
\usepackage{amsthm}
\usepackage{latexsym}
\usepackage[dvips]{epsfig}

\usepackage{wasysym}
\usepackage{mathrsfs}
\usepackage{eufrak}
\usepackage{bm}
\usepackage{authblk}
\usepackage{slashed}
\usepackage{yhmath} 
\usepackage{stmaryrd}

\theoremstyle{plain}
\newtheorem{proposition}{Proposition}
\newtheorem{lemma}{Lemma}
\newtheorem{theorem}{Theorem}
\newtheorem*{theoremnonumber}{Theorem}
\newtheorem{assumption}{Assumption}

\newtheorem{corollary}{Corollary}
\newtheorem*{main}{Main Result}
\newtheorem{definition}{Definition}
\newtheorem{remark}{Remark}
\newtheorem*{condition}{Condition}

\newtheorem*{notation}{Notation}

\def\er{{\bar{r}}}
\def\bh{\bar{h}}
\def\bD{\bar{D}}

\def\bK{\bar{K}}

\def\bma{{\bm a}}

\def\bmf{{\bm f}}
\def\bmg{{\bm g}}
\def\bmh{{\bm h}}

\def\bml{{\bm l}}

\def\bmu{{\bm u}}
\def\bmv{{\bm v}}

\def\bmA{{\bm A}}

\def\mcS{\mathcal{S}}
\def\mcM{\mathcal{M}}

\def\me{\mathrm{e}}
\def\mi{\mathrm{i}}

\def\msD{D^{\perp}}

\newcounter{mnotecount}

\newcommand{\mnotex}[1]
{\protect{\stepcounter{mnotecount}}$^{\mbox{\footnotesize $\bullet$\themnotecount}}$ 
\marginpar{
\raggedright\tiny\em
$\!\!\!\!\!\!\,\bullet$\themnotecount: #1} }

\usepackage[utf8]{inputenc}
\usepackage[english]{babel}
 
\usepackage{multicol}
\usepackage{color}
 
\usepackage{comment}
 
\begin{document}

\title{\textbf{Dain's invariant for black hole initial data}}

\author[1]{R. Sansom\footnote{E-mail address:{\tt
      r.sansom@qmul.ac.uk}}} \author[,1]{J. A. Valiente
  Kroon \footnote{E-mail address:{\tt j.a.valiente-kroon@qmul.ac.uk}}}

\affil[1]{School of Mathematical Sciences,
  Queen Mary, University of London, Mile End Road, London E1 4NS,
  United Kingdom.}

\maketitle

\begin{abstract}
Dynamical black holes in the non-perturbative regime are not
mathematically well understood. Studying approximate symmetries of
spacetimes describing dynamical black holes gives an insight into
their structure. Utilising the property that approximate symmetries
coincide with actual symmetries when they are present allows one to
construct geometric invariants characterising the symmetry. In this
paper, we extend Dain's construction of geometric invariants
characterising stationarity to the case of initial data sets for the
Einstein equations corresponding to black hole spacetimes. We prove
the existence and uniqueness of solutions to a boundary value problem
showing that one can always find approximate Killing vectors in black
hole spacetimes and these coincide with actual Killing vectors when
they are present. In the time-symmetric setting we make use of a 2+1
decomposition to construct a
geometric invariant on a MOTS that vanishes if and only if the Killing
initial data equations are locally satisfied.
\end{abstract}

\medskip

\section{Introduction}
For a generic initial data set of the vacuum Einstein field equations
an important question is whether the development of this initial data
set possesses a Killing vector or symmetry. This question first arose
in the study of linearisation stability \cite{Mon75} and is of utmost
importance as symmetries greatly simplify problems in all fields
within General Relativity. Moreover, some of the most important
solutions to the Einstein field equations admit a number of
symmetries. In the context of the initial value problem of General
Relativity, the existence of continuous symmetries is characterised by
solutions of the Killing initial data (KID) equations \cite{BeiChr97a}
---see also \cite{Mon75}. The KID equations are a system of
overdetermined equations for a scalar and a 3-vector on the initial
data set such that, if a solution exists, the development of the data
will have a Killing vector with lapse and shift at the initial
hypersurface given by the scalar and vector, respectively. In fact,
these equations also have a deep connection with the ADM evolution
equations and the adjoint linearised constraint map $D\Phi^*$, as described in Section \ref{AKEsection}.

\medskip
As already mentioned, the condition that a spacetime $(\mcM, \bmg)$ possesses
a Killing vector is encoded in the initial data by the Killing initial data (KID)
equations. This holds, in particular, for time translations. Given the
role that stationary solutions have in the mathematical description of
isolated bodies, it is natural to attempt to quantify the deviation of a given initial
data set from that of a stationary one. In \cite{Dai04c}, the notion of
an {\it approximate Killing vector} was introduced as a solution to a
fourth order, self-adjoint, elliptic partial differential
equation whose solution implies the existence of a geometric invariant
which allows one to identify static initial data. This idea was developed and extended in \cite{ValWil17} to
the non-time symmetric case ---thus allowing for a characterisation of
stationary data sets. In these works the approximate Killing vector
equation was solved on the whole of an asymptotically Euclidean
3-manifold. That is, the only boundary conditions prescribed in this
analysis were those associated to the asymptotic ends. As such, this
approach does not allow to encode the presence of a black hole.

\medskip
The purpose of this article is to extend the results of \cite{Dai04c}
and \cite{ValWil17} to incorporate not only asymptotic conditions but
also an inner boundary that corresponds to the boundary of a black
hole. This boundary will be a
2-dimensional surface that encodes the presence of a black hole in the
evolution of the initial development. This is accomplished by the
condition that the surface is a marginally outer trapped surface (MOTS),
sometimes referred to as an apparent horizon.

\medskip
The main result of this paper is contained in the following theorem. 
\begin{main}
Let $(S,h_{ij},K_{ij})$ be a complete, smooth asymptotically Euclidean
initial data set for the Einstein vacuum field equations with one
asymptotic end. Given smooth fields $f$, $g$, $f^i$ and $h^i$ over
$\partial\mathcal{S}$, then there exists a solution $(X,X^i)$ to the
approximate KID equation boundary value problem
\begin{align*}
&\mathscr{P}\circ\mathscr{P}^{*}
\begin{pmatrix}
X\\
X^{i}
\end{pmatrix}
=0 \qquad \text{on } S,\\
&\begin{cases}
X\vert_{\partial S} =f,\\
\Delta_{h} X \vert_{\partial S} =g,\\
X^{i}\vert_{\partial S}=f^i,\\
\pdv{\rho}X^{i}\vert_{\partial S}=h^i,
\end{cases}
\end{align*}
such that on the asymptotic end, one has the asymptotic behaviour
\begin{align*}
&X=\lambda |x|+\vartheta_{} \qquad \vartheta\in H^{\infty}_{1/2}\\
&X^{i} \in  H^{\infty}_{1/2}
\end{align*}
where $\lambda_{}$ is Dain's invariant ---i.e. if the data is
stationary in the sense of Definition \ref{defstat} then
$\lambda_{}$ vanishes.
\end{main}
The generality of the boundary values in this theorem is exploited in
order to choose physically relevant values of the fields. The construction of
the boundary values is clearly non-unique due to the arbitrariness of
the fields. We consider the decomposition of the Killing initial data
(KID) equations onto $\partial\mcS$ and investigate how much of the
KID equations one can solve on this surface. Integral to this
construction is the use of the MOTS stability operator
\cite{AndMarSim08}. In this manner, in the time symmetric setting, we
incorporate the condition that the MOTS propagates into the evolution
of the initial data in the boundary conditions of the approximate KID
operator. However, in the non-time symmetric case, how one should
incorporate the stability of MOTS is much less obvious.

\medskip
In addition to the main theorem above, in the time symmetric setting,
we also construct a geometric invariant on the surface $\partial\mcS$
that vanishes if and only if the KID equations are solved there. This
result is contained in the following theorem and proved in Section
\ref{KIDsolvabilitysec}.

\begin{theoremnonumber}
  Given a time symmetric initial data set and a MOTS the KID equations give rise to the elliptic equation on $\partial\mcS$ 
  \[
  \Delta_{\bh}N-\frac{1}{2}(\er+\bK^{pq}\bK_{pq}) N=0.
  \]
 For a stable MOTS the unique solution is $N=0$. Furthermore, using the evolution equations one obtains
\[
\Delta_{h}N=N=0, \qquad \mbox{on} \quad \partial\mathcal{S}.
\]
Then the KID equations are satisfied on $\partial
\mathcal{S}$ if and only if $\omega=0$ where
\begin{equation*}
\omega \equiv \int_{\partial\mathcal{S}} |\bK|^{2}|\msD N|^{2},
\end{equation*}
with $|\bK|^{2}=\bK_{pq}\bK^{pq}$ and $D^{\perp}\equiv \rho^{i}D_{i}$.
\end{theoremnonumber}

\subsection*{Outline of the article}
The paper is organised as follows: in Section \ref{AKEsection} the
approximate Killing initial data equation (AKE) is constructed as in
\cite{ValWil17} and \cite{Dai04c}. In order to make sure that a
suitable boundary value problem is constructed, natural boundary
operators are constructed using the Green formula. The compatibility
of the boundary operators and the AKE are then verified using the
Lopatinskij-Shapiro condition. Section
\ref{Section:MainExistenceTheorem} comprises the main result of this
article, namely, the existence and uniqueness of the AKE with inner
boundary conditions. In Section \ref{KIDeqnMOTSsec}, we investigate
and derive the decomposition of the KID equations into components
tangential and normal to the inner boundary surface. Finally, in
section \ref{KIDsolvabilitysec}, we prove existence of solutions to
the tangential KID equations and give a construction of a geometric
invariant that vanishes only when one has a Killing vector. In
Appendix \ref{gfappendix} we present the Green formula for the
approximate KID operator, in Appendix \ref{LSODEappendix} we show the
full derivation of the solution to the ODE arising from the
Lopatinski-Shapiro condition and in Appendix
\ref{Appendix:DecompsotionKIDs} we derive the decomposed components of
the KID equations normal and tangential to $\partial\mcS$.

\subsection*{Notations and Conventions}
The indices, $a,b,c,...$ are spacetime indices, $i,j,k,...$ are
indices on an initial 3-dimensional hypersurface, $\mcS$ with metric
$h_{ij}$ and unit normal $n^{a}$. The indices $A,B,C,...$ are indices
on a 2-dimensional surface $\Sigma$ with metric $\bh_{AB}$ and unit normal $\rho^{i}$ embedded in the 3-dimensional surface. The induced covariant derivative on
$\mcS$ is $D_{i}$ and on $\Sigma$ is $\bD_{A}$ We use the positive
convention on the extrinsic curvature, that is,
$K_{ab}=+h_{a}{}^{c}h_{b}{}^{d}\nabla_{c}n_{d}$ for the extrinsic
curvature of $\mcS$ and
$\bK_{ij}=+\bh_{i}{}^{k}\bh_{j}{}^{l}D_{k}\rho_{l}$. In the latter
sections of the paper, we will make use of Gaussian normal coordinates
on 2-dimensional surfaces, that is the acceleration of the foliation
vanishes ---in other words, $a_{i}=0$.

\section{The approximate Killing vector equation}\label{AKEsection}
In this section, we develop the theory of the approximate Killing vector equation
and introduce some compatible boundary operators. The first part of
this section was developed in \cite{ValWil17}.

\subsection{The approximate Killing operator}
Denote an initial data set of the vacuum Einstein field equations by
$(\mathcal{S},h_{ij},K_{ij})$ where $\mathcal{S}$ is a 3-dimensional
manifold, $h_{ij}$ is a Riemannian metric on $\mathcal{S}$ and
$K_{ij}$ is a symmetric rank 2 tensor satisfying the vacuum Einstein
constraint equations
\begin{subequations}
\begin{eqnarray}
&& r+ K^2 -K_{ij} K^{ij} =0, \label{HamiltonianConstraint}\\
&& D^j K_{ij} - D_i K=0. \label{MomentumConstraint}
\end{eqnarray}
\end{subequations}
We refer to equations \eqref{HamiltonianConstraint} and
\eqref{MomentumConstraint} as the \emph{Hamiltonian} and
\emph{momentum} constraints, respectively. In the above expressions
$D_i$ denotes the Levi-Civita connection of the metric $h_{ij}$, $r$
is the associated Ricci scalar and $K\equiv K_{ij} h^{ij}$.

\medskip
In the sequel, we will be particularly interested in initial data sets
whose development is endowed with a Killing vector. At the level of
the initial data set, this property is encoded through the existence
of a solution to the Killing initial data (KID) equations.

\begin{proposition}\label{prop:kid}
{\em (Killing initial data (KID) equations.)} Let
$(\mathcal{S},h_{ij},K_{ij})$ be an initial data set for the vacuum
Einstein field equations. If there exists a pair $(N,Y^i)$ such that
\begin{subequations}
\begin{eqnarray}
&& NK_{ij}+D_{(i}Y_{j)} =0, \label{kid1}\\
&& N^{k}D_{k}K_{ij}+D_{i}N^{k}K_{kj}+D_{j}N^{k}K_{ik}+D_{i}D_{j}N=N\left(r_{ij}+KK_{ij}-2K_{ik}K^{k}_{j}\right), \label{kid2}
\end{eqnarray}
\end{subequations}
then the development of the initial data is endowed with a Killing
vector and $(N,Y^i)$ are the lapse and shift of this Killing vector at
$\mathcal{S}$.
\end{proposition}

A proof of this result can be found in \cite{BeiChr97b}.

\medskip
Following \cite{ValWil17}, we write the constraint equations
\eqref{HamiltonianConstraint} and \eqref{MomentumConstraint} as a map,
$\Phi:
\mathfrak{M}_2\cross\mathfrak{T}_2\rightarrow\mathfrak{S}\cross\mathfrak{X}$,
where $\mathfrak{M}_2$ is the space of 3-dimensional Riemannian
metrics, $\mathfrak{T}_2$ is the space symmetric 2-tensors,
$\mathfrak{S}$ the space of scalars and $\mathfrak{X}$ is the space of
vectors on $\mathcal{S}$:
\[
\Phi\begin{pmatrix}
h_{ij}\\
K_{ij}
\end{pmatrix}
=
\begin{pmatrix}
 r+ K^2 -K_{ij} K^{ij} \\
 D^j K_{ij} - D_i K
\end{pmatrix}.
\]
Linearising and finding the formal adjoint of this linearisation through integration by parts yields
\begin{equation}
D\Phi^*\left(\begin{array}{c}
X\\
X_i
\end{array}\right)=\left(\begin{array}{c}
D_iD_jX-X r_{ij}-\Delta_h X h_{ij}+H_{ij}\\
D_{(i}X_{j)}-D^kX_k h_{ij}+F_{ij}
\end{array}\right)
\end{equation}
where $H_{ij}$ and $F_{ij}$ are as in \cite{ValWil17}.

\begin{remark}
{\em A calculation shows that having a solution $(X,X_i)$ to
  $D\Phi^*(X,X_i)=0$ is equivalent to $(X, X_i)$ satisfying the KID
  equations ---see e.g. Remark 2 in \cite{ValWil17}.}
\end{remark}

The above remark gives the motivation behind the following definition:
\begin{definition}
For the operator $\mathscr{P}\circ\mathscr{P}^*:\mathfrak{S}\cross\mathfrak{X}\rightarrow\mathfrak{S}\cross\mathfrak{X}$, the equation
\begin{equation}\tag{AKE}\label{AKE}
\mathscr{P} \circ \mathscr{P}^{*}
\begin{pmatrix}
X\\
X_{i}
\end{pmatrix}
=0
\end{equation}
where this operator is given by 
\begin{equation*}
\hspace{-5mm}\mathscr{P}\circ\mathscr{P}^*\left(\begin{array}{c}
X\\
X_i
\end{array}\right) \equiv \left(\begin{array}{c}
2\Delta_\bmh\Delta_\bmh X-r^{ij}D_iD_jX+2r\Delta_\bmh X+\tfrac{3}{2}D^i rD_iX+(\tfrac{1}{2}\Delta_\bmh r+r_{ij}r^{ij})X\\
+D^iD^j H_{ij}-\Delta_\bmh H_k{}^k-r^{ij}H_{ij}+\bar{H} \\[1em]
D^j\Delta_\bmh D_{(i}X_{j)}+D_i\Delta_\bmh D^kX_k+D^j\Delta_\bmh F_{ij}-D_i\Delta_\bmh F_k{}^k-\bar{F}_i
\end{array}\right)
\end{equation*}
with
\begin{align*}
& \bar{H}\equiv 2(K\bar{Q}-K^{ij}\bar{Q}_{ij})+2(K^{ki}K^j{}_k-KK^{ij})\bar{\gamma}_{ij},\\
& \bar{F}_i\equiv \left(D_i K^{kj}-D^k K^j{}_i\right)\bar{\gamma}_{jk}-\left(K^k{}_i D^j-\tfrac{1}{2}K^{kj}D_i\right)\bar{\gamma}_{jk}+\tfrac{1}{2}K^k{}_i D_k \bar{\gamma}\\
&\bar{\gamma}_{ij}\equiv D_iD_jX-X r_{ij}-\Delta_\bmh X h_{ij}+H_{ij}\\
&\bar{Q}_{ij}\equiv -\Delta_\bmh(D_{(i}X_{j)}-D^kX_k h_{ij}+F_{ij})\\
&H_{ij}\equiv 2X(K^k{}_iK_{jk}-KK_{ij})-K_{k(i}D_{j)}X^k+\tfrac{1}{2}K_{ij}D_kX^k
+\tfrac{1}{2}K_{kl}D^kX^l h_{ij}-\tfrac{1}{2}X^kD_k K_{ij}+\tfrac{1}{2}X^k D_k K h_{ij}\\
& F_{ij}\equiv 2X(Kh_{ij}-K_{ij})
\end{align*}
is {\em the approximate Killing vector equation (AKE)}. A solution $(X,X_{i})$ to this equation is called an {\em approximate Killing vector}.
\end{definition}

\begin{remark}
{\em Following from Dain, \cite{Dai04c}, every approximate Killing vectors
corresponds to one of the ten Killing vectors in flat spacetime. In
the following, we focus on the approximate killing vectors that
correspond to the Killing vector associated with time translations in
flat spacetime. This analysis could be extended to the other Killing
vectors corresponding to spatial translations, boosts and rotations.}
\end{remark}

\begin{remark}
{\em When the initial data is {\em time symmetric}, that is $K_{ij}=0$, the AKE simplifies to
\begin{equation}\label{aKID}
\mathscr{P} \circ \mathscr{P}^{*}
\begin{pmatrix}
X\\
X_{i}
\end{pmatrix}
=
\begin{pmatrix}
2\Delta_{h}\Delta_{h}X - r^{ij}D_{i}D_{j}X+\frac{1}{2}r_{ij}r^{ij}X\\
D^{j}\Delta_{h}D_{(i}X_{j)}+D_{i}\Delta_{h}D^{k}X_{k}
\end{pmatrix}
\end{equation}
where we have used the constraint equations in this case to set $r=0$.
Note that under the assumption of time symmetry, these equations
decouple from one another and can thus be considered as two separate
equations: one for the lapse $X$ and one for the shift $X_i$.}
\end{remark}

 We have the following important properties of the AKE operator:
\begin{lemma}
The operator $\mathscr{P}\circ\mathscr{P}^*$ as defined above is a self adjoint, fourth order elliptic operator.
\end{lemma}
A proof of this result can by found in in \cite{ValWil17}. In order to
discuss the solvability of the \eqref{AKE}, we will need to introduce weighted
Sobolev spaces and the notion of an asymptotically Euclidean
manifolds.

\subsection{The Approximate Killing vector equation on asymptotically Euclidean manifolds}
In this section, we summarise the results of \cite{ValWil17} for the
solvability of the AKE on asymptotically Euclidean manifolds. We make
use of weighted Sobolev spaces to discuss the decay of various tensor
fields on the 3-dimensional manifold $\mcS$.

\subsubsection{Weighted Sobolev spaces and asymptotically Euclidean manifolds}
We begin with the definition of a weighted Sobolev space, $H_{\delta}^{s}$:
\begin{definition}
Given points $p,x \in \mcS$, let $s$ be a non-negative integer and
$\delta \in \mathbb{R}$. The {\em weighted Sobolev space} denoted by
$H_{\delta}^{s}$ consists of of all functions, $u$, such that the {\em
weighted Sobolev norm} is finite
\begin{equation*}
    ||u||_{s,\delta}\equiv\sum_{0\leq |\alpha|\leq s}||D^\alpha u||_{\delta - |\alpha|}<\infty
\end{equation*}
where $\alpha=(\alpha_1, \alpha_2, \alpha_3)$ is a multiindex and the norm in the summand is the weighted $L^2$-norm 
\begin{equation*}
    ||\phi||_\delta \equiv \left(\int_{\mathcal{S}} |\phi|^2\sigma^{-2\delta -3}\mathrm{d}^3 x\right)^{1/2}
\end{equation*}
with $\sigma(x)\equiv (1+d(p,x)^2)^{1/2}$ and $d$ denotes the Riemannian distance on $\mcS$. One says that $u\in H_{\delta}^{\infty}$ if $u\in H_{\delta}^{s}$ for all $s$.
\end{definition}

\begin{remark}
{\em We follow Bartnik's conventions \cite{Bar86} for the weighted Sobolev
spaces and norms. Note also that we are slightly abusing notation
since these norms seem to be dependent on $p$. However, different
choices of $p$ give rise to equivalent weighted Sobolev norms, see
e.g. \cite{Bar86,ChoChr81a}. Thus, we denote these norms with the same
symbol.}
\end{remark}

Using these weighted Sobolev spaces, we are in a position to discuss
the necessary fall-off conditions in order to form an asymptotically
Euclidean manifold. Consider an initial data set $(\mcS, h_{ij},
K_{ij})$ for the Einstein vacuum field equations that has $n$ {\em
asymptotically Euclidean ends}. That is, there exists a compact set
$\mathcal{B}$ such that
\[
    \mcS \setminus \mathcal{B} = \sum_{k=1}^{n}\mcS_{(k)}
\]
where $\mcS_{(k)}$ are open sets diffeomorphic to the complement of a
closed ball in $\mathbb{R}^3$. Each $\mcS_{(k)}$ is called an {\em
asymptotic end}. On each of these ends one can introduce {\em
asymptotically Cartesian coordinates} $x=(x^\alpha)$. The definition
of an asymptotically Euclidean manifold is then defined on these ends.

\begin{definition}\label{asympeuclidmf}
The 3-dimensional manifold $\mcS$ is called {\em asymptotically
Euclidean} if on each asymptotic end one has that
\begin{align*}
    &h_{\alpha\beta}-\delta_{\alpha\beta}\in H_{-\frac{1}{2}}^{\infty},\\
    &K_{\alpha\beta} \in H_{-\frac{3}{2}}^{\infty}.
\end{align*}
\end{definition}

\subsubsection{Green's Formula and the Fredholm alternative}

Green's formula will be fundamental to the choice of boundary operators in constructing the boundary value problem for the \eqref{AKE}. In this section, we outline the basic theory of Green's formula and use this to motivate the Fredholm alternative. The latter will be necessary in proving the main theorem.
\\\\
Green's formula of an elliptic differential operator $\mathscr{A}$ arises when
considering the formal adjoint $\mathscr{A}^*$ of $\mathscr{A}$ \cite{Wlo87, Dai06}. Let
$\Omega$ be a bounded smooth domain in $\mathbb{R}^n$. For all
${\bf u},{\bf v}$ compactly supported in $\Omega$, the {\em formal adjoint} is
defined through
\begin{equation*}
\int_\Omega {\bf v}\cdot\mathscr{A}{\bf u}\mathrm{d}\mu=\int_\Omega {\bf u}\cdot\mathscr{A}^{*}{\bf v}\mathrm{d}\mu.
\end{equation*}
The adjoint is calculated by performing integration by parts. If we
now consider ${\bf u},{\bf v}$ to be not compactly supported on $\Omega$,
performing integration by parts yields boundary terms. The resulting
relation is known as {\it Green's formula} for $\mathscr{A}$
\begin{equation*}
\int_\Omega \left({\bf v}\cdot\mathscr{A}{\bf u}- {\bf u}\cdot\mathscr{A}^{*}{\bf v}\right)\mathrm{d}\mu = \oint_{\partial \Omega} {\bf v}\cdot\mathscr{B}{\bf u} -{\bf u}\cdot\mathscr{T}{\bf v}\mathrm{d}S
\end{equation*}
where $\mathscr{T}$ and $\mathscr{B}$ are {\it differential boundary
operators}. In order to formalise this discussion, we first define a
Dirichlet system.
\begin{definition}
Let $\Gamma$ be a subset of $\partial\Omega$. The boundary value
operators $b_{j}^{\alpha}(x,D)$, $j=1,...,n$ and $\alpha=1,...,N$ is
the number of equations, form a {\em Dirichlet system} on $\Gamma$ if
and only if
\begin{itemize}
\item[i.] The order, $m_{j}^{\alpha}$, of $b_{j}^{\alpha}$ is such that $m_{i}^{\alpha}\neq m_{j}^{\alpha}$ for $i\neq j =1,...,n$,
\item[ii.] The symbol of the operator $\sigma_{j}(x,\vec{\xi})\neq 0$ $\forall$ $x\in\Gamma$ and $\vec{\xi}\neq 0$ and is normal to $\partial \Omega$ at $x$,
\item[iii.] For each $\alpha$, the orders $m_{j}^{\alpha}$ run through all numbers $0,1,...,n-1$ (without loss of generality $m_{j}^{\alpha}=j-1$). The number $n$ is called the {\em order of the Dirichlet system}.
\end{itemize}
\end{definition}
A set of boundary value operators satisfying only points 1 and 2 above
is said to be {\em normal}. Then Green's formula can be expressed as
the following:
\begin{proposition}
Let $\mathscr{A}(x,D)$ be an elliptic operator on $\bar{\Omega}$ and
$b_{j}^{\alpha}(x,D)$, $j=1,...,m$, $\alpha=1,...,N$, be a normal
boundary value system on $\partial\Omega$. Then on $\partial\Omega$
one can find another boundary value system $S_{j}^{\alpha}$,
$j=1,...,m$, $\alpha=1,...,N$, with orders $\mu_{j}^{\alpha}<2m-1$ so
that $\{b_{1}^{1},...,b_{m}^{N},S_{1}^{1},...,S_{m}^{N}\}$ is a
Dirichlet system of order $2mN$ on $\partial\Omega$.\footnote{The
choice of $S_{j}^{\alpha}$ is not unique.} Additionally, one can
construct a further $2mN$ boundary value operators
$B_{j}'^{\alpha},T_{j}^{\alpha}$, $j=1,...,m$ with the properties:
\begin{itemize}
\item[i.] The orders of $B_{j}'^{\alpha}$ and $T_{j}^{\alpha}$ are given by $2m-1-\mu_{j}^{\alpha}$ and $2m-1-m_{j}^{\alpha}$, respectively.
\item[ii.] The set  $\{B_{1}'^{1},...,B_{m}'^{N},T_{1}^{1},...,T_{m}^{N}\}$ is also a Dirichlet system of order $2m$ on $\partial \Omega$.
\item[iii.] We have Green's formula:
\begin{equation}\label{greensformula}
\int_\Omega \left({\bf v}\cdot\mathscr{A}{\bf u}- {\bf u}\cdot\mathscr{A}^{*}{\bf v}\right)\mathrm{d}\mu = \sum_{j=1}^{m}\sum_{\alpha=1}^{N}\oint_{\partial \Omega}\left( S_{j}^{\alpha}{\bf u}\cdot B'^{\alpha}_{j}{\bf v} -b^{\alpha}_{j}{\bf u}\cdot T_{j}^{\alpha}{\bf v}\right)\mathrm{d}S.
\end{equation}
\end{itemize}
\end{proposition}

Thus, the operators arising in Green's formula are the natural
boundary operators to consider. The above discussion generalises to
operators over a manifold.

\medskip
In the sequel, we will need to make use of the Fredholm alternative
for elliptic operators acting between weighted Sobolev spaces
which relies on the asymptotic homogeneity of the approximate Killing
operator. We outline the notion of asymptotic homogeneity here. In
local coordinates on $\mcS$, the \eqref{AKE} can be written as
\begin{equation*}
    \mathscr{L}\bmu \equiv (\boldsymbol{A}^{\alpha\beta\gamma\delta}+\boldsymbol{a}^{\alpha\beta\gamma\delta})\cdot\partial_{\alpha}\partial_{\beta}\partial_{\gamma}\partial_{\delta}\bmu+\bma^{\alpha\beta\gamma}\cdot\partial_{\alpha}\partial_{\beta}\partial_{\gamma}\bmu +\bma^{\alpha\beta}\cdot\partial_{\alpha}\partial_{\beta}\bmu + \bma^\alpha\cdot\partial_{\alpha}\bmu + \bma\cdot\bmu =0,
\end{equation*}
where $\bmu:\mcS \rightarrow \mathbb{R}^4$ is a vector valued function
over $\mcS$, $\bmA^{\alpha\beta\gamma\delta}$ are a constant matrices,
while
$\bma^{\alpha\beta\gamma\delta},\bma^{\alpha\beta\gamma},\bma^{\alpha\beta},\bma^{\alpha}$
and $\bma$ are smooth matrix-valued functions of the coordinates
$(x^\alpha)$. Then $\mathscr{L}$ is {\em asymptotically homogeneous}
if the matrix-valued functions belong to the following weighted
Sobolev spaces
\[
    \bma^{\alpha\beta\gamma\delta}\in H_{\tau}^{\infty},\quad \bma^{\alpha\beta\gamma}\in H_{\tau-1}^{\infty}, \quad \bma^{\alpha\beta}\in H_{\tau-2}^{\infty},\quad
    \bma^{\alpha}\in H_{\tau-3}^{\infty}, \quad
    \bma\in H_{\tau-4}^{\infty},
\]
for $\tau<0$, see \cite{Can81, Loc81}. With this definition, we can
classify the asymptotic homogeneity of the approximate Killing
operator in local coordinates.

\begin{lemma}
If the 3-dimensional manifold $\mcS$ is asymptotically Euclidean as in
Definition \ref{asympeuclidmf} then $\mathscr{L}$ is asymptotically
homogeneous with $\tau=-1/2$.
\end{lemma}

We will make use of the following form of the Fredholm
alternative, as proved in \cite{Wlo87} (see also \cite{Can81}):

\begin{proposition}\label{prop:fredholmalternative}
 Let $\mathscr{L}$ be a fourth order asymptotically homogeneous elliptic
operator over a smooth domain $\Omega$ with smooth coefficients and
let $b_i^{\alpha}$, $i=1,2$, $\alpha=1,...,N$, be smooth boundary operators on $\partial
\Omega$. Given some non-negative integer $\delta$, the boundary value
problem
\begin{align*}
 \begin{cases}
\mathscr{L}{\bmu}={\bmf} \qquad &\text{ in } \Omega\\
b_1^{1}{\bmu}={\bmg}^{1}_1\\
\vdots \qquad &\text{ on } \partial \Omega\\
 b_2^{N}{\bmu}={\bmg}^{N}_2
 \end{cases}
\end{align*}
where ${\bmf, \bmg_1^{1},...,\bmg_2^{N}}\in H^{0}_{\delta -4}$ possesses at least one solution $\bmu\in H^4_\delta$ if and only if 
\[
\int_{\Omega} {\bmf}\cdot {\bmv}\mathrm{d}\mu + \sum_{j=1}^{2}\sum_{\alpha=1}^{N}\int_{\partial\Omega}{\bmg_{j}}^{\alpha}\cdot{T_{j}^{\alpha}\bmv}\mathrm{d}\sigma=0 \qquad \forall \bmv\in N^*
\]
where $T_i$ are the boundary operators coming from Green's
formula and $N^*$ is the space of solutions to the homogeneous adjoint
boundary value problem
\begin{align*}
 \begin{cases}
\mathscr{L}^*{\bmv}=0 \qquad &\text{ in } \Omega\\
B'^{1}_1{\bmv}=0\\
\vdots \qquad &\text{ on } \partial \Omega\\
 B'^{N}_2{\bmv}=0
 \end{cases}
\end{align*}
for $\bmv\in H^0_{1-\delta}$. $\mathscr{L}^*$ denotes the formal adjoint of $\mathscr{L}$.
\end{proposition}

\subsubsection{Existence of solutions to the AKE on asymptotically Euclidean manifolds}
The notion of approximate Killing vectors and approximate symmetries
was introduced by Dain \cite{Dai04c}. For time symmetric initial data
it was shown that:
\begin{itemize}
    \item[a.] every Killing vector is also an approximate Killing vector;
    \item[b.] for generic initial data that admits no Killing vector (i.e. no solution to the KID equations) there always exists an approximate Killing vector;
    \item[c.] every approximate Killing vector can be uniquely associated with a Killing vector in flat spacetime.
\end{itemize}

An invariant denoted by $\lambda_{(k)}$ was also constructed on each
asymptotic end such that $\lambda_{(k)}=0$ if and only if the initial
data is stationary in the sense of the following definition:

\begin{definition}\label{defstat}
An asymptotically Euclidean initial data set $(\mcS, h_{ij}, K_{ij})$
is called {\em stationary} if there exists non-trivial $(Y,Y_i)\in
H_{\frac{1}{2}}^2$ such that
\[
    \mathscr{P}^{*}
\begin{pmatrix}
Y\\
Y_{i}
\end{pmatrix}
=0.
\]
Moreover, if the initial data is also time symmetric, i.e. $K_{ij}=0$
then, if the above condition holds, the initial data is called {\em
static}.\footnote{Note that this condition is equivalent to the KID
equations due to the fall off on $(Y,Y_i)$. See Remark (6) in
\cite{ValWil17}.}
\end{definition}

The above results were extended in \cite{ValWil17} to the non-time
symmetric setting. We state the main theorem (Theorem 1) of this work
here, without proof:

\begin{theorem}
Let $(\mathcal{S},h_{ij},K_{ij})$ be a complete, smooth
asymptotically Euclidean initial data set for the Einstein vacuum
field equations with $n$ asymptotic ends. Then there exists a 
solution $(X,X^i)$ to the AKE,
\[
\mathscr{P}\circ \mathscr{P}^* 
\left(
\begin{array}{c}
X \\
X^i
\end{array}
\right) =0,
\]
such that at each asymptotic end one has the asymptotic
behaviour
\begin{eqnarray*}
 && X_{(k)} = \lambda_{(k)} |x| + \vartheta_{(k)}, \qquad \vartheta_{(k)} \in    H^\infty_{\frac{1}{2}},\\
&& X^i_{(k)} \in H^\infty_{\frac{1}{2}},
 \end{eqnarray*}
where $\lambda_{(k)}$, $k=1,\ldots,n$, are constants and $\lambda_{(k)}=0$ for some $k$
if and only if $(\mathcal{S},h_{ij},K_{ij})$ is stationary in the
sense of Definition \ref{defstat}. Moreover, the solution is unique up to constant rescaling.
\end{theorem}

\begin{remark}
{\em In the following work, we restrict to one asymptotic end with one inner boundary. This can be extended to one asymptotic end with multiple inner boundaries corresponding to spacetimes with multiple black holes.}
\end{remark}

The goal of the present work is to extend this theorem to include
inner boundary conditions on $\mcS$. To do this, it is important that
the constructed boundary value problem is solvable and that the
boundary 2-dimensional surface represents that of a black hole
boundary. In this context the natural surface to consider is a
marginally outer trapped surface (MOTS), sometimes referred to as an
{\em apparent horizon}. The remainder of this section will tackle
constructing a solvable boundary value problem. In particular, we will
compute Green's formula in order to obtain natural
boundary operators and verify that the Lopatinskij-Shapiro condition
holds for the AKE with these natural boundary operators. Thereby
showing that this boundary value problem is Fredholm.

\begin{remark}
  {\em As a general principle, for an elliptic boundary value
problem, the derivative order of the boundary condition has to
be lower than the operator and the number of boundary conditions must
be half the order of the equation. The AKE is comprised of 4 fourth-order equations. Thus, we must have 8 boundary conditions.}
\end{remark}

\subsection{Green's formula for the AKE}
\medskip
In this section, we derive Green's formula for the AKE using equation \eqref{greensformula}. The derivation is essentially
integration by parts and thus, due to the size of some of the terms,
the calculations can be found in full in Appendix \ref{gfappendix}.  

\medskip
By inspecting the calculation in Appendix A, we can construct Green's
formula for the (AKE). We work here using the components of the
equation instead of vectorial quantities as in equation 
\eqref{greensformula}. Thus, the obtained Dirichlet systems will have
in total 16 elements since the \eqref{AKE} has four components.

\begin{lemma}\label{lemgfAKE}
The Green formula for the AKE can be written as
\begin{align*}
\int_{\mathcal{S}}\mathscr{P} \circ \mathscr{P}^{*}
\begin{pmatrix}
X\\
X_{i}
\end{pmatrix}
\cdot
\begin{pmatrix}
Z\\
Z_{i}
\end{pmatrix}
&-
\int_{\mathcal{S}}\mathscr{P} \circ \mathscr{P}^{*}
\begin{pmatrix}
Z\\
Z_{i}
\end{pmatrix}
\cdot
\begin{pmatrix}
X\\
X_{i}
\end{pmatrix}\\
&=\\
\sum_{j=1}^{2}\sum_{\alpha=1}^{4} \left(\oint_{\partial \mcS}S^{\alpha}_{j}
(X,X_{i})
\cdot B^{'\alpha}_{j}(Z,Z_i)\right.&-\left.\oint_{\partial \mcS} b^{\alpha}_{j}(X,X_i)\cdot T^{\alpha}_{j}(Z,Z_i)\right).
\end{align*}
 Thus, one has the Dirichlet systems $\{b_1^{1},...,b_2^{4},S_1^{1},...,S_2^{4}\}$ and $\{B^{'1}_1,...,B^{'4}_2,T_1^{1},...,T_2^{4}\}$.
\end{lemma}
A useful property that elliptic boundary value problems can have is 
self-adjointness. That is, if the $b_j$ appearing in the formula in
Lemma \ref{lemgfAKE} are the operators appearing in a boundary value
problem then the $B^{'}_j$ are the {\em adjoint boundary
operators}. If $b_j^{\alpha}=B^{'\alpha}_j$ then the boundary value problem
$(L,b_1^{1},...,b_2^{4})$ is {\em self-adjoint} for $L=L^*$. Since we have
that the AKE is self adjoint, it will prove incredibly useful to
consider a self-adjoint boundary value problem ---particularly, when
employing the \emph{Fredholm alternative}.

\begin{corollary}\label{cornatbo}
One can choose the boundary operators $b_j^{\alpha}$ to be  
\begin{equation}\label{nbc}
\left\{I,\Delta_{h}X,\displaystyle\pdv{\rho}X_{i}\right\}.
\end{equation}
where $I$ is the identity operator acting on $(X,X_i)$.
This choice yields a self-adjoint boundary value problem for the AKE.
\end{corollary}

\begin{remark}
{\em The final operator in this set is found by decomposing terms of the
form $D_iX_j$ into tangential and normal components i.e.
\begin{equation*}
D_{i}X_{j} = h_{i}^{k}D_{k}X_{j}=(\bh_{i}^{k}+\rho_{i}\rho^{k})D_{k}X_{j}=\bh_{i}^{k}D_{k}X_{j}+\rho_{i}\rho^{k}D_{k}X_{j}.
\end{equation*}
The tangential term is determined by $X_i$, leaving only the normal
derivative to be determined. We also note the change in direction of
the unit normal. The reason for this choice is so that the normal
vector is normal to what will become a MOTS.}
\end{remark}

Next, we check that these boundary operators are compatible with the
approximate KID equation. To do this, we make use of the
Lopatinskij-Shapiro condition.

\subsection{Verifying the Lopatinskij-Shapiro condition}
\label{Subsection:LS}
The Lopatinskij-Shapiro (LS) condition allows us to establish the
compatibility of an elliptic operator $\mathscr{L}$, with some boundary
operators $\mathscr{B}$. That is, if $(\mathscr{L},\mathscr{B})$ satisfy the LS condition then
$(\mathscr{L},\mathscr{B})$ is Fredholm i.e. its kernel and cokernel are finite
dimensional. We will give a brief overview of the LS condition and
then prove that it holds for the AKE with the boundary operators
derived in corollary \ref{cornatbo}. For more details see
\cite{Wlo87}.

\medskip
Let $u^A$, $A=1,...,N$ be a collection of fields on a subset
$\Omega\subseteq \mathbb{R}^n$ with coordinates $x=(x^\alpha)$ and
suppose we have $N$ equations of at most order $l$. We consider
operators of the form
\[
(\mathscr{L}u)_i = \sum_{0\leq|\gamma|\leq l} L^\gamma_{iB}(x^{\alpha},u)\partial_\gamma u^{B}
\]
where $i=1,...,n$ are equation indices, $\gamma$ is a multiindex. We
complement $\mathscr{L}$ with boundary operators on $\partial\Omega$ of the form
\[
    (\mathcal{B}u)_j=\sum_{0\leq|\gamma|\leq k} B^\gamma_{jB}(x^{\alpha},u)\partial_\gamma u^{B}
\]
with $j=1,...,m$ so that there are $m$ boundary conditions.

\medskip
In order
to state the LS condition, we need the definition of the principal
part of a differential operator. Recall that to obtain the {\em principal part} of an operator
$\mathscr{L}$ we consider only the highest order derivative terms in
the operator ---namely
\[
    (\mathscr{L}^H u)_i = \sum_{|\gamma| = l} L^\gamma_{iB}(x^{\alpha},u)\mathcal{D}_\gamma u^{B} \equiv A^{\gamma_1...\gamma_l}_{iB}\partial_{\gamma_1}...\partial_{\gamma_l}u^B,
  \]
  where the Einstein summation convention is used in the second
  equality. In particular, the principal part of the AKE is 
\[
    (\mathscr{P} \circ \mathscr{P}^{*})^H
\begin{pmatrix}
X\\
X_{i}
\end{pmatrix}
=
\begin{pmatrix}
2\Delta_{\delta}\Delta_{\delta}X \\
\partial^{j}\Delta_{\delta}\partial_{(i}X_{j)}+\partial_{i}\Delta_{\delta}\partial^{k}X_{k}
\end{pmatrix}.
\]
An important detail which we will exploit in the following lemma is that the components of the principal part of the AKE operator decouple from one another and can thus be considered separately.
We now are in a position to state the LS condition 
\begin{condition}
{\em(Lopatinskij-Shapiro (LS)).} Focus only on the principal parts of
$\mathscr{L}$ and $\mathscr{B}$. Let $x_*\in\partial\Omega$ and pick
$x^1=\rho$ such that $\rho^i$ is the outward pointing normal to
$\partial\Omega$. Consider the ODE problem
\begin{align*}
    \begin{cases}
   \mathscr{L}_{iB}(x_*;\dv{}{\rho},\vec{\xi})u^B=0,\\
    \mathscr{B}_{jB}(x_*;\dv{}{\rho},\vec{\xi})u^B=0,
    \end{cases}
\end{align*}
where $ A_{iB}(x_*;\dv{}{\rho},\vec{\xi})$ is obtained by the replacement \begin{equation*}
\partial_{1}\rightarrow \dv{}{\rho}, \qquad \partial_{j}\rightarrow \mathrm{i}\xi_{j} \qquad  j=2,...,n,\quad \vec{\xi}\neq 0
\end{equation*}
in the principal part:
$A^{\gamma_1...\gamma_l}_{iB}\partial_{\gamma_1}...\partial_{\gamma_l}$
and similarly for $B_{jB}(x_*;\dv{}{\rho},\vec{\xi})$.  Then, the pair
$(\mathscr{L},\mathscr{B})$ is said to satisfy the LS condition of the
only stable solution\footnote{That is, $u^{B}\rightarrow 0$ as $\rho \rightarrow \infty$.} to the above ODE is the trivial one.
\end{condition}

It is important to note that the LS condition is verified about an
arbitrary point on the boundary. Thus, we can generalise to $\mcS$ and
$\partial\mcS$. We now verify the LS condition for the AKE and
boundary conditions given in Lemma \ref{lemgfAKE}.

\begin{lemma}\label{lslemma}
For an initial data set $(\mcS,h_{ij},K_{ij})$ of the vacuum Einstein
field equations, with inner boundary $\partial\mathcal{S}$, the boundary
value problem consisting of the \eqref{AKE}
\[
\mathscr{P} \circ \mathscr{P}^{*}
\begin{pmatrix}
X\\
X_{i}
\end{pmatrix}
 = \begin{pmatrix}
0\\
0
\end{pmatrix}
\]
along with the boundary operators $(I,\Delta_{h},\pdv{\rho})$ given on
$\partial\mcS$ satisfy the Lopatinskij-Shapiro condition.
\end{lemma}

\begin{proof}
Due to the decoupling of the principal parts of the AKE one can
consider the lapse and shift components independently. Thus, begin by considering the
principal part of the lapse component of the AKE operator
\[
\mathscr{P} \circ \mathscr{P}^{*}(X)=2\Delta_{\delta}\Delta_{\delta}X.
\]
with the the boundary operators on $\partial\mcS$ given by $(I,\Delta_{\delta})$. Focusing on the principal part of the lapse component of the AKE and the principal part of the boundary operators, consider the ordinary differential equation problem given by
\begin{align*}
\begin{cases}
\left(\displaystyle\dv[2]{}{\rho}-|\xi|^{2}\right)\left(\displaystyle\dv[2]{}{\rho}-|\xi|^{2}\right)X=0,\\
\Delta_{\delta}X=0,\\
X=0,
\end{cases}
\end{align*}
where this has been obtained by choosing a point on $\partial\mcS$ and performing the replacement
\begin{equation*}
\partial_{1}\rightarrow \dv{}{\rho}, \qquad \partial_{j}\rightarrow \mathrm{i}\xi_{j}.
\end{equation*}
Since one is performing the replacement about a point, one also has
that $h_{ij}\rightarrow\delta_{ij}$. The stable solution is given by
\begin{equation*}
X=c_{1}e^{-|\xi|\rho}+c_{2}\rho e^{-|\xi|\rho}.
\end{equation*}
Then, applying the boundary conditions above at $\rho=0$, one sees that
$c_{1}=c_{2}=0$ and thus the solution is trivial and the
Lopatinskij-Shapiro condition is satisfied for the lapse component of
the AKE with the above boundary operators.

\medskip
The shift component of the AKE has principal part given by
\[
\partial^{j}\Delta_{\delta}\partial_{(i}X_{j)}+\partial_{i}\Delta_{\delta}\partial^{k}X_{k}.
\]
Unlike the case of the lapse component of
the AKE there are three equations in this case corresponding to
$i=1,2,3$. Thus, one requires six boundary conditions. Due to the
derivatives either side of the Laplacian, carrying out the
transformation in order to verify the LS condition yields a system of
ODEs with terms containing derivatives of fourth order and lower.

Using the commutativity of partial derivatives one finds that, multiplying through by 2, the principal part is
\[
\Delta_\delta(\Delta_\delta X_{i}+3\partial_{i}\partial^{j}X_{j}).
\]
Making use of this expression and performing the replacement
$X_{i}\rightarrow (X^\perp,X_{A})$ with $A=1,2$ to this expression directly
yields the system of ordinary differential equations
\[
\begin{cases}
\left(\displaystyle\dv[2]{}{\rho}-|\xi|^{2}\right)\left(4 \displaystyle\dv[2]{}{\rho}X^\perp -|\xi|^{2}X+3\mi\xi^{A}\displaystyle\dv{}{\rho}X_{A}\right)=0,\\
\left(\displaystyle\dv[2]{}{\rho}-|\xi|^{2}\right)\left(\left(\displaystyle\dv[2]{}{\rho}-|\xi|^{2}\right)X_{A}+3\mi\xi_{A}\left(\displaystyle\dv{}{\rho}X^\perp+\mi\xi^{B}X_{B}\right)\right)=0.
\end{cases}
\]
Analysing the fundamental matrix of this system one finds that the solution to the above system of ODEs is of the form
\[
  X_{i}=\sum_{k=0}^{2} X_{*i}\rho^{k}\me^{\pm |\xi| \rho},
\]
Following a computation (see Appendix \ref{LSODEappendix}) one finds that the stable solution is given by
\[
\begin{cases}
X= a\me^{-|\xi| \rho}+ {\displaystyle\frac{\mi}{|\xi|}}b_{A}\xi^{A}\rho\me^{-|\xi| \rho}+c(\frac{3}{10}|\xi|\rho^{2}+\rho)\me^{-|\xi| \rho},\\
X_{A}=a_{A}\me^{ -|\xi| \rho}+ b_{A}\rho\me^{ -|\xi| \rho}-c\frac{3}{10}\mi\xi_{A}\rho^{2}\me^{ -|\xi| \rho}.
\end{cases}
\]
One then requires 6 boundary conditions to fix the 6 constants
corresponding to $a,a_{A},b_{A},c$. Consider the principal parts of
the boundary operators $X_{i}$ and $\pdv{\rho} X_{i}$. Performing the
Lopatinskij-Shapiro replacements yields the initial conditions
\[
\begin{cases}
X(0)=0,\\
X_{A}(0)=0,\\
\displaystyle\dv{}{\rho}X(0)=0,\\
\displaystyle\dv{}{\rho}X_{A}(0)=0.
\end{cases}
\]
Substituting the stable solution into these conditions yields
$a=a_{1}=a_{2}=b_{1}=b_{2}=c=0$ and thus this boundary value problem
has only the trivial solution and the Lopatinskij-Shapiro condition is
satisfied by the shift component of the AKE along with the boundary
operators $(I,\pdv{\rho})$.

\medskip
Thus, we have shown that the Lopintskij-Shapiro condition is satisfied
for the 8 boundary operators $(I,\Delta_{h}N,\pdv{\rho}N_{i})$.
\end{proof}

Since the LS condition is satisfied we have the following corollary:
\begin{corollary}\label{cor:akefredholm}
The approximate Killing vector operator 
$
    \mathscr{P} \circ \mathscr{P}^{*}
\begin{pmatrix}
X\\
X_{i}
\end{pmatrix}
$
along with the 8 boundary operators $\{I, \Delta_h X, \pdv{}{\rho}\}$
is elliptic.
\end{corollary}

\section{Main existence theorem}
\label{Section:MainExistenceTheorem}
We are now in a position to prove a series of existence results to the
\eqref{AKE}. We begin with an auxiliary existence result which will be
essential to prove our main existence result.

\begin{lemma}\label{vanishingbdrylemma}
Let $(S,h_{ij},K_{ij})$ be a complete, smooth asymptotically Euclidean
initial data set for the Einstein vacuum field equations with one
asymptotic end. On $\partial\mcS$, suppose one has the
conditions
\begin{equation*}
\begin{cases}
N\vert_{\partial \mcS} =0,\\
\Delta N \vert_{\partial \mcS} =0,\\
N^{i}\vert_{\partial \mcS}=0,\\
\pdv{\rho}N^{i}\vert_{\partial \mcS}=0.
\end{cases}
\end{equation*}
 For $0<\beta<\frac{1}{2}$, $\mathscr{P}\circ\mathscr{P}^{*}:H_{\beta}^{\infty}\rightarrow H_{\beta-4}^{\infty}$, the equation
\begin{equation*}
\mathscr{P}\circ\mathscr{P}^{*}
\begin{pmatrix}
N\\
N^{i}
\end{pmatrix}
=0,
\end{equation*}
admits a non-zero solution $N,N^{i} \in H_{\beta}^{\infty}$ if and only if we have that
\begin{equation*}
\mathscr{P}^{*}
\begin{pmatrix}
N\\
N^{i}
\end{pmatrix}
=0,
\end{equation*}
i.e. the data is stationary in the sense of Definition \ref{defstat}.
\end{lemma}

\begin{proof}
Assume that $\mathscr{P}\circ\mathscr{P}^{*}(N, N^{i})=0$. One has the
following identity from \cite{ValWil17}
\begin{align*}
\int_{\mcS}\mathscr{P}^*\left(
\begin{array}{c}
X \\
X_i
\end{array}
\right) \cdot
\mathscr{P}^*
\left(
\begin{array}{c}
X \\
X_i
\end{array}
\right) = \oint_{\partial\mcS} \rho^k \big( \mathcal{A}_k + \mathcal{B}_k + \mathcal{C}_k + \mathcal{D}_k \big)
\mbox{d}S-\oint_{\partial\mathcal{S}_{\infty}} s^k \big( \mathcal{A}_k + \mathcal{B}_k + \mathcal{C}_k + \mathcal{D}_k \big)
\mbox{d}S
\end{align*}
where $S_{\infty}$ is a sphere at infinity, $\rho^{k}$ is the {\em
inward} pointing normal to $\partial\mcS$ and $s^{k}$ is the outward
pointing normal to the sphere at infinity. The second boundary
integral vanishes by virtue of the decay of $N,N^{i}$ as shown in
\cite{ValWil17}. The boundary integrands are defined as
\begin{align*}
 \mathcal{A}_{k} &= N D^i\gamma_{ik} - D^iN \gamma_{ik} + D_k N\gamma -ND_k\gamma,\\
 \mathcal{B}_{k} &=2(K^{ij}q_{kij}-Kq_{kj}{}^{j})N, \\
 \mathcal{C}_{k} &=N^i D^l q_{lik} - D^j N^i q_{kij} + D_i N^i q_{kj}{}^j -N_i D^l q_{lj}{}^j,\\
 \mathcal{D}_{k} &= \frac{1}{2}N_k K^{lj} \gamma_{jl} + \frac{1}{2}N^i K_{ik}\gamma - N^i K_i{}^l \gamma_{kl},\\
 \gamma_{ij} &= D_i D_j N - N r_{ij} -\Delta_{h} N h_{ij} + H_{ij},\\
 q_{kij} &= D_k\big( D_{(i} N_{j)} - D^l N_l h_{ij} - F_{ij}\big),\\
 F_{ij} &=2N (K h_{ij} - K_{ij}).
\end{align*} 
Directly applying the boundary conditions, one observes immediately
that 
\[
  \mathcal{B}_{k}=\mathcal{D}_{k}=F_{ij}=0 \quad \mbox{on} \quad \partial \mcS.
  \]

\medskip
Using that $N^{i}=0$, $\mathcal{C}_{k}$ reduces to 
\begin{equation*}
 \mathcal{C}_{k} =- D^j N^i q_{kij} + D_i N^i q_{kj}{}^j.
\end{equation*}
Using the decomposition of the metric onto the boundary we have that
\begin{equation*}
D_{i}X_{j} = h_{i}^{k}D_{k}X_{j}=(\bh_{i}^{k}+\rho_{i}\rho^{k})D_{k}X_{j}=\bh_{i}^{k}D_{k}X_{j}+\rho_{i}\rho^{k}D_{k}X_{j}.
\end{equation*}
The first term on the right hand side vanishes since $X_{i}$ is
constant and thus does not change along the boundary, and the second
term vanishes by the boundary condition $\rho^{k}D_{k}X_{j}=0$. Thus,
$\mathcal{C}_{k}=0$. One also readily sees that the trace of $\gamma_{ij}$
vanishes. Accordingly, the only term left is the second one in $\mathcal{A}_{k}$:
$\rho^{k}D^{i}N\gamma_{ik}$. This reduces to
$\rho^{k}D^{i}ND_{i}D_{k}N$ as all other terms in $\gamma_{ij}$
vanish. We have that, using the condition on the Laplacian of $N$, it
follows that
\begin{equation*}
0 = h^{ij}D_{i}D_{j}N=\bh^{ij}D_{i}D_{j}N +\rho^{i}\rho^{j}D_{i}D_{j}N.
\end{equation*}
The first term on the right hand side vanishes as the derivatives
tangential to $\partial \mcS$ of $N$ vanish. Thus, one obtains the
identity
\begin{equation*}
\rho^{j}\pdv{}{\rho}D_{j}N=0.
\end{equation*}
Then, since the 3-dimensional covariant derivative is assumed to be torsion free, we can write 
\begin{equation*}
\rho^{k}D^{i}ND_{i}D_{k}N=\rho^{k}D^{i}ND_{k}D_{i}N=D^{i}N\pdv{}{\rho}D_{i}N = \rho^{i}\rho_{k}D^{k}N\pdv{}{\rho}D_{i}N=0
\end{equation*}
applying the previous identity.

\medskip
 Putting this all together yields
\begin{equation*}
\int_{\mathcal{U}}\mathscr{P}^*\left(
\begin{array}{c}
X \\
X_i
\end{array}
\right) \cdot
\mathscr{P}^*
\left(
\begin{array}{c}
X \\
X_i
\end{array}
\right) =0,
\end{equation*}
so that we must have $\mathscr{P}^*\left(
\begin{array}{c}
X \\
X_i
\end{array}
\right)=0$ , i.e. the data is stationary. Uniqueness then follows
directly from Lemma 4 of \cite{ValWil17} since we have the same decay
on $(N,N^{i})$.
\end{proof}

We are now able to prove that a solution to the \eqref{AKE} exists
for an arbitrary choice of boundary data.

\begin{theorem}
Let $(\mcS,h_{ij},K_{ij})$ be a complete, smooth asymptotically Euclidean
initial data set for the Einstein vacuum field equations with one
asymptotic end. Given  smooth fields $f$, $g$, $f^i$ and $h^i$ over
$\partial\mathcal{S}$, then there exists a solution $(X,X^i)$ to the
approximate KID equation boundary value problem
\begin{align*}
&\mathscr{P}\circ\mathscr{P}^{*}
\begin{pmatrix}
X\\
X^{i}
\end{pmatrix}
=0 \qquad \text{on } S,\\
&\begin{cases}
X\vert_{\partial S} =f,\\
\Delta_{h} X \vert_{\partial S} =g,\\
X^{i}\vert_{\partial S}=f^i,\\
\pdv{\rho}X^{i}\vert_{\partial S}=h^i,
\end{cases}
\end{align*}
such that on the asymptotic end, one has the asymptotic behaviour
\begin{align*}
&X_{(k)}=\lambda_{}|x|+\vartheta_{}, \qquad \vartheta_{} \in H^{\infty}_{\frac{1}{2}}\\
&X^{i}_{} \in  H^{\infty}_{\frac{1}{2}}
\end{align*}
where $\lambda_{}$ is Dain's invariant ---i.e. if the data is
stationary in the sense of Definition \ref{defstat}--- then
$\lambda_{}$ vanishes.
\end{theorem}
\begin{proof}
We make use of the Fredholm alternative, Proposition
\ref{prop:fredholmalternative}. Note first that the boundary value
problem is self adjoint by virtue of Corollary \ref{cornatbo}. From
Lemma \ref{vanishingbdrylemma}, in the non-stationary case we can only have the
trivial solution to the adjoint homogeneous problem. Thus, a solution
to the AKE boundary value problem always exists in this case.

\medskip
In the stationary case, we have a solution to the AKE boundary value
problem for specific boundary data coming from the Killing vector
associated to stationarity. The decay of $(X,X^i)$ allows one to
prove existence and uniqueness exactly as in \cite{ValWil17}.
\end{proof}

In summary, we have shown that there always exists an approximate
Killing vector for an arbitrary choice of the functions specifying the
value of the lapse, its Laplacian, the shift and its normal derivative
on the inner boundary. Ideally, we would like to restrict the choice
of data so that it has physical relevance and connects with the
mathematical description of a black hole. In order to do this we
analyse the tangential parts of the KID equations on a MOTS. In the
sequel, we will show that one can always solve these equations on the
boundary. Thus, we derive a natural prescription of boundary data.


\section{The KID equations on an apparent horizon}\label{KIDeqnMOTSsec}
In this section we decompose the KID equations into parts tangential
and normal to a MOTS using a 2+1 projector formalism. This works in the same way as the 3+1 decomposition. We
first outline a condition on the extrinsic curvature of $\mcS$ and
$\partial\mathcal{S}$ for the presence of an apparent horizon in
\ref{secmotscondition}. In section \ref{sec12decomp} we decompose the
KID equations using a `2+1 decomposition' making use of the MOTS
condition.

\subsection{MOTS condition}\label{secmotscondition}
In this section we outline the mathematical condition for the
existence of a MOTS or apparent horizon, see \cite{Cho08}. The
spacelike 2-surface $\Sigma$ is embedded into a 4-dimensional
spacetime. In this setting, the orthogonal space to $\Sigma$ is
2-dimensional and is Lorentzian. Thus, one can choose two future
directed null vectors $\bml^{+}$ and $\bml^{-}$ and define the null mean
curvatures of $\Sigma$ by
\[
\chi^{+}=\bh^{ab}\nabla_{a}l^{+}_{b},\qquad \chi^{-}=\bh^{ab}\nabla_{a}l^{-}_{b}.
\]
The MOTS condition is then defined through the vanishing of one of
these null mean curvatures ---by convention this is chosen to be $\chi^{+}=0$. One
can translate the above formalism into a condition on the extrinsic
curvature of $\mcS$ and $\Sigma$ as in \cite{Cho08} to obtain 
\[
\chi^{+}=-K+\rho^{i}\rho^{j}K_{ij}+\bK.
\]
Thus, the MOTS condition can be expressed as
\begin{equation}\label{motscondition}
0=-K+\rho^{i}\rho^{j}K_{ij}+\bK.
\end{equation}

Using this equation, we will be able to encode into the presence of a
black hole into the boundary data of the \eqref{AKE}. The notion of
MOTS stability \cite{AndMarSim05,AndMarSim08} will be utilised in
order to guarantee that the MOTS propagates into the development of
the initial data. In other words, so that the MOTS forms a so-called
{\em trapping horizon}. Fundamental to the study of MOTS stability, is
the MOTS stability operator ---see \cite{AndMarSim08} for details. In
the case of a spacelike, time-symmetric initial hypersurface the MOTS
stability operator takes the form
\begin{equation}\label{motsstabilityoperator}
 \mathcal{L}=-\Delta_{\bh}-(r_{ij}\rho^{i}\rho^{j}+\bK_{pq}\bK^{pq}).
\end{equation}

\subsection{2+1 Decomposition}\label{sec12decomp}
In this subsection we discuss the $2+1$ decomposition of the KID
equations.
We first make a brief comment about notation.
\begin{notation}
{\em Throughout this section all barred quantities correspond to
2-dimensional quantities. For example, $h_{ij}$ is 3-dimensional and $\bh_{AB}$
is 2-dimensional. Following the convention in \cite{ChoDeWDil82}, we
denote the decomposition of a vector field at a point $p$ by
$\bmu=\bmu^{\parallel}+\bmu^{\perp}$ where $\bmu^{\parallel}$ is its
component along $T_{p}\Sigma$ for a 2-surface $\Sigma$ and
$\bmu^{\perp}$ its normal component. This extends to higher rank
tensor fields.}
\end{notation}

In the following consider a 3-dimensional hypersurface $\mcS$ of the
spacetime. Following the conventions outlined above let $\bmh$ denote
the metric induced on $\mcS$. Moroever, let $\Sigma$ denote a
2-dimensional surface within $\mcS$. Let $\varsigma$ be a scalar function such that the covector $\alpha_{i}=D_{i}\varsigma$ is normal to this foliation. In order to define the unit normal, let us set 
\begin{equation*}
\alpha^{i}\alpha_{i}\equiv \frac{1}{X^{2}}.
\end{equation*}
Thus, the {\it unit normal to the foliation} is $\rho_{i}\equiv
X\alpha_{i}$. By applying the above contraction, one sees
$\rho_{i}\rho^{i}=1$. In this way we obtain a foliation of 2-surfaces,
one 2-surface for each value of $\varsigma$.

\medskip
The Riemannian $3$-metric $h_{ij}$ induces a Riemannian $2$-metric
$\bar{h}_{AB}$ on $\Sigma_{r}$, where indices $A,B,...$ indicate the
intrinsic $2$-dimensional nature of $\bar{h}_{AB}$. The metrics
$h_{ij}$ and $\bh_{AB}$ are related through the projector
\[
\bh_{ij}=h_{ij}-\rho_{i}\rho_{j}
\]
One associates a $2$-covariant derivative with the metric
$\bh_{AB}$. For a scalar this is defined as
\[
\bD_{j}\phi\equiv \bh_{j}{}^{i}D_{i}\phi.
\]
The associated Riemann curvature tensor is defined through
\[
\bD_{i}\bD_{j}v^{k}-\bD_{j}\bD_{i}v^{k}=\er\indices{^{k}_{lij}}v^{l}
\]
for an intrinsic vector $v^{A}$. Note that since this curvature tensor is $2$-dimensional, it only has one independent component. Similarly, one obtains the Ricci tensor and Ricci scalar
\[
\er_{ij}=\er\indices{^{k}_{ikj}}\qquad\er=\bh^{ij}\er_{ij}.
\]
The {\it extrinsic curvature} of the $2$-surface embedded in the 3-dimensional surface is defined by
\[
\bK_{pq}=\bh\indices{_{p}^{i}}\bh\indices{_{q}^{j}}D_{i}\rho_{j}
\]
matching the positive sign convention on $K_{ij}$. This tensor is
symmetric and lies purely on the $2$-dimensional surfaces. Another
important quantity is the acceleration of the foliation 
\[
a_{i}=\rho^{j}D_{j}\rho_{i}.
\]
Using $\rho^{i}D_{j}\rho_{i}=0$ one can write
\[
\bK_{ij}=D_{i}\rho_{j}-\rho_{i}a_{j}.
\]
\begin{remark}
{\em In the following, we will often make use of {\em Gaussian normal
  coordinate} so that this acceleration vanishes, $a_i=0$. }
\end{remark}

Now, one can project the KID equations~\eqref{kid1}-\eqref{kid2} onto
the $2$-surfaces of constant $\varsigma$. Let ${N^{\parallel}}^{i}$ denote the projection of
the shift vector $N^i$ onto the $2$-surfaces ---i.e. 
${N^{\parallel}}^{i}\equiv \bh\indices{^{i}_{j}}N^{j}$. We proceed now
to decompose objects into parts perpendicular and parallel to the
normal $\rho^{i}$. Projecting with $$\bh\indices{_{p}^{i}}\bh\indices{_{q}^{j}}$$ we obtain the projection of the the first KID
equation~\eqref{kid1}
\begin{equation}\label{fulldecompkideq1}
NK^{\parallel}_{pq}+\bD_{(p}N^{\parallel}_{q)}+N_{k}\rho^{k}\bK_{pq}=0,
\end{equation}
where
$K^{\parallel}_{pq}\equiv\bh\indices{_{p}^{i}}\bh\indices{_{q}^{j}}K_{ij}$. 

\medskip
For the second KID equation \eqref{kid2}, we project term by term making use of the Gauss-Codazzi equation and Gaussian normal coordinates to obtain the projection of the second KID equation onto $\Sigma$:
\begin{align}\label{fulldecompkideq2}
\begin{split}
&N^{k}\Bigl( D_{k}K_{pq}^{\parallel}+K_{iq} \bigl(\rho^{i} (\bK_{kp}) + (\bK_{k}{}^{i}) \rho_{p}\bigr) + K_{pj} \bigl(\rho^{j} (\bK_{kq}) + (\bK_{k}{}^{j} ) \rho_{q}\bigr)\\
 &- K_{ij} \bigl(\rho^{j} (\bK_{k}{}^{i}) \rho_{p} \rho_{q} + \rho^{i} (\bK_{k}{}^{j} ) \rho_{p} \rho_{q} + \rho^{i}\rho^{j} (\bK_{kq}) \rho_{p} +  \rho^{i}\rho^{j} (\bK_{kp} ) \rho_{q}\bigr)\Bigr) \\
 &+\bh\indices{_{p}^{i}}\bh\indices{_{q}^{j}}D_{i}N^{k}K_{kj}+\bh\indices{_{p}^{i}}\bh\indices{_{q}^{j}}D_{j}N^{k}K_{ik} + \bD_{p}\bD_{q}N+\rho^{k}D_{k}N\bK_{pq}\\
 &=N\left( \bh\indices{_{p}^{i}}\bh\indices{_{q}^{j}}r_{ij}+KK^{\parallel}_{pq}+\bh\indices{_{p}^{i}}\bh\indices{_{q}^{j}}K_{ik}K\indices{^{k}_{j}}\right).
 \end{split}
\end{align}

\medskip
Together, equations \eqref{fulldecompkideq1} and
\eqref{fulldecompkideq2} give conditions on $\Sigma$ that Killing
vectors of the the spacetime must satisfy on this 2-surface. However,
the converse is not true. If $(N,N^i)$ a solution to these two
equations that does not necessarily mean that one has a solution to
the KID equations. This is because there are two other components of
the KID equations in this decomposition which we call the
normal-normal and normal-tangential components. For the first KID equation
\eqref{kid1}, one has
\begin{alignat*}{2}
\mathrm{normal-normal:}&\qquad NK^{\perp}_{pq} &&+\rho_{(p\vert}\rho^{j}D_{j\vert}N^{\perp}_{q)}-N_{i}a^{i}\rho_{p}\rho_{q}-N^{\perp}_{(p}a_{q)}=0,\\
\mathrm{normal-tangential:}& \qquad NK_{pq}^{\perp\parallel}&&+\frac{1}{2}\left(\rho_{p}\rho^{i}D_{i}N^{\parallel}_{q} +N_{j}a^{j}\rho_{p}\rho_{q}+N^{\perp}_{p}a_{q}\right)\\
& &&+\frac{1}{2}\left(\bh_{q}{}^{j}D_{j}N^{\perp}_{p}-N_{i}\rho_{p}\bK_{q}{}^{i}-\rho^{i}N_{i}\bK_{pq}\right)=0,\nonumber
\end{alignat*}
where $K_{pq}^{\perp}\equiv\rho_{p}\rho^{i}\rho_{q}\rho^{j}K_{ij}$ and
$K_{pq}^{\perp\parallel}\equiv
\rho_{p}\rho^{i}\bh_{j}{}^{q}K_{ij}$. For the second KID equation
\eqref{kid2}, one can obtain the projection under the assumption of
Gaussian coordinates. The normal-normal and normal-tangential components of the second KID equations are
\begin{align*}
\begin{split}
\mathrm{normal-normal:}\qquad &\rho_{p}\rho_{q}N^{k}D_{k}(K-\bK)-\rho_{p}\rho_{q}N^{k}K_{ij}(\rho^{i}\bK_{k}{}^{j}+\rho^{j}\bK_{k}{}^{i})\\
+&\rho_{p}\rho_{q}\rho^{i}\rho^{j}(D_{i}N^{k}K_{kj}+D_{j}N^{k}K_{ik})+ \rho_{p}\rho^{i}D_{i}(\rho_{q}
\rho^{j}D_{j}N)\\
=&N(-\bh_{p}{}^{i}\bh_{q}{}^{j}r_{ij}+KK^{\perp}_{pq}+\rho_{p}\rho_{q}\rho^{i}\rho^{j}K_{ik}K_{j}^{k}),\\\\
\mathrm{normal-tangential:} \qquad &h_{p}{}^{i}\rho_{q}\rho^{j}(N^{k}D_{k}K_{ij}+D_{i}N^{k}K_{kj}+D_{j}N^{k}K_{ik})\\
+& \rho_{p}\rho^{i}D_{i}\bD_{q}N\\
=&N(h_{p}{}^{i}\rho_{q}\rho^{j}r_{ij}+KK^{\perp\parallel}+h_{p}{}^{i}\rho_{q}\rho^{j}K_{i}{}^{k}K_{kj}).
\end{split}
\end{align*}

\subsubsection{Time symmetric hypersurfaces}
In order to gain some intuition into the structure of the decomposition
of the KID equations we consider, in first instance, the case of time
symmetric hypersurfaces so that the extrinsic curvature $K_{ij}$
vanishes. Under this assumption, the trapping condition simplifies to
$\bK=0$ i.e the surface $\Sigma$ is minimal. Moroever, the
decomposition of the first KID equation implies that 
\begin{equation}
\bD_{p}{N^{\parallel}}^{p}=0.\label{kiddc1}
\end{equation}
after taking the trace. 
\begin{remark}
{\em Due to the decay condition placed on $N^i$ in the previous
  section, in the time symmetric setting one has that $N^i=0$
  \cite{ChrOMu81,Dai04c}. By performing
  integration by parts on the time symmetric \eqref{AKE}, one finds that $N^i=0$ on
  $\partial\mathcal{S}$. Direct inspection of the KID equations also
  shows that $N^{i}$ must be a Killing vector of $\mcS$ in the time
  symmetric case.}
\end{remark}

In the current setting the decomposition of the second KID reduces to
\[
\bD_{p}\bD_{q}N=N\bh\indices{_{p}^{i}}\bh\indices{_{q}^{j}}r_{ij}.
\]
Taking the trace of the above expression one has that 
\[
\Delta_{\bh}N=N\bh^{ij}r_{ij}=\frac{1}{2}N(\er+\bK^{pq}\bK_{pq}),
\]
where in the second equality it has been used that
\[
\er+\bK^{ij}\bK_{ij}-\bK^{2}=2\bh^{ij}r_{ij},
\]
which is a consequence of the Gauss-Codazzi equation. Accordingly, the
decomposition of the second KID equation implies that 
\begin{equation}
\Delta_{\bh}N-\frac{1}{2}(\er+\bK^{pq}\bK_{pq}) N=0.\label{kiddc2}
\end{equation}
\begin{remark}
{\em This equation has a very similar form to the MOTS stability operator given in \cite{AndMarSim08}.}
\end{remark}

One can perform the same simplification on the normal-normal and normal-tangential components of the
equations. Additionally, using Gaussian normal coordinates one obtains
\begin{subequations}
\begin{alignat}{2}
\mathrm{normal-normal:}&\qquad \rho_{p}\rho^{i}D_{i}(\rho_{q}\rho^{j}D_{j}N)=N(h_{p}{}^{i}h_{q}{}^{j}-\bh_{p}{}^{i}\bh_{q}{}^{j})r_{ij} \label{nncompkid}\\
\mathrm{normal-tangential:}& \qquad  \rho_{p}\rho^{i}D_{i}\bD_{q}N=N \rho_{p}\rho^{i}\bh_{q}{}^{j}r_{ij}\label{nscompkid}
\end{alignat}
\end{subequations}

\section{The KID equations projected onto $\partial \mathcal{S}$}\label{KIDsolvabilitysec}
In this section, we study the existence of solutions to the decomposed
KID equations on an apparent horizon. We first study the
simpler time symmetric case, $K_{ij}=0$ in Subsection
\ref{sec:timesymkid} and move on to the full non-time symmetric case
in Subsection \ref{generalcasesolvability}. 

\subsection{The time symmetric case}\label{sec:timesymkid}
Throughout this section we assume that $K_{ij}=0$. We begin with some
general observations.

\subsubsection{What do Killing vector quantities look like on MOTS in
  a static spacetime?}
Before proceeding to the analysis, we briefly explore the behaviour of
the quantities $X,\Delta_{h}X,X_{i}$ and $\pdv{}{n}X_{i}$ on the
boundary of a static black hole. The
prototypical example of a static black hole is the Schwarzschild black
hole. In this case the outermost MOTS is the event horizon
\cite{BauSha10}. Thus, we only need to consider what happens at $r=2m$
in Schwarzschild coordinates. Now, we need to choose an appropriate
slicing. Choosing the standard Schwarzschild slicing, the shift vector vanishes everywhere by
definition so that $X_{i}=\pdv{}{n}X_{i}=0$. The lapse is
$X=(1-\frac{2m}{r})^{1/2}$ so that on the event horizon
$X=0$. Finally, the Laplacian of the lapse vanishes by virtue of one
of the the static Einstein equations ---namely,  $D_{i}D_{j}X=Xr_{ij}$. Taking the trace and
using that $X=0$ on the event horizon, we obtain
$\Delta_{h}X=0$. Accordingly, the type of boundary conditions given in
Corollary \ref{cornatbo} are physically relevant.

\subsubsection{Existence and uniqueness of solutions}

We now want to discuss the existence of solutions to equation
\eqref{kiddc2}. Recall that this equation 
arises from the  2+1 decomposition of the second KID equation,
equation along the 2-surface $\partial\mathcal{S}$. Letting  
---namely
\begin{equation}\label{bkid1}
\mathcal{K}N\equiv \Delta_{\bh}N-\frac{1}{2}(\er+\bK_{pq}\bK^{pq}) N=0.
\end{equation}

Starting from the quantity $N\mathcal{K}N$, integrating over $\partial\mathcal{S}$
and integrating by parts, one finds that 
\begin{align}\label{IbPkid2}
  \oint_{\partial\mathcal{S}}
\frac{1}{2}(\er+\bK_{pq}\bK^{pq})N^{2}\mathrm{d}S=\oint_{\partial\mathcal{S}}N\Delta_{\bh}N\mathrm{d}S = -\oint _{\partial\mathcal{S}}\bD_{i}N\bD^{i}N\mathrm{d}S.
\end{align}
Thus, if $\er>-\bK_{pq}\bK^{pq}$, in particular if $\er>0$ then $N=0$
is the only solution to equation \eqref{bkid1}. To further this
analysis, assume that the MOTS is stable. In this case, since we have
the MOTS is a surface immersed in a time symmetric slice, the
stability operator takes exactly the form of \eqref{motsstabilityoperator}.
 Using the Gauss-Codazzi equation the expression \eqref{motsstabilityoperator} can be rewritten as
\[
\mathcal{L}=-\Delta_{\bh}+\frac{1}{2}(\er- \bK_{pq}\bK^{pq}). 
\]

\begin{remark}
{\em Note the remarkable similarity with the equation for the
tangential component of the second KID equation above
\eqref{bkid1}. In fact, notice that
\[
\mathcal{K} = \mathcal{L} + |\bar{K}|^2.
\]}
\end{remark}

Making use of the above observations one obtains the following lemma:

\begin{lemma}
Assuming the MOTS is stable, the only solution to \eqref{bkid1} is the trivial one, i.e. $N=0$.
\end{lemma}

\begin{proof}
In the following let $\lambda$ and $\mu$ denote, respectively, the lowest positive eigenvalues
of $\mathcal{L}$ and $\mathcal{K}$. We proceed to compare $\lambda$
and $\mu$. For this, we make use of the Rayleigh-Ritz characterisation
of these eigenvalues. Namely, one has that 
\begin{eqnarray*}
&& \lambda = \inf_u \oint_{\partial\mathcal{S}} \Big(|\bar{D}u|^2
   +\tfrac{1}{2}(\bar{r} -|\bar{K}|^2)u^2\Big)\mathrm{d}S ,\\
&& \mu = \inf_u \oint_{\partial\mathcal{S}} \Big(|\bar{D}u|^2
   +\tfrac{1}{2}(\bar{r} +|\bar{K}|^2)u^2\Big)\mathrm{d}S, 
\end{eqnarray*}
where the infimum is take over functions $u$ on $\partial\mathcal{S}$
with $|| u ||^2_{L^2}=1$. It then follows that $\lambda \leq \mu$.
Thus, if $\partial\mathcal{S}$ is a stable MOTS, then $\lambda>0$ and
accordingly $\mu>0$. Now, using Lemma 4.2, (iii) in \cite{AndMarSim08}
it follows that the Kernel of $\mathcal{K}$ is trivial. Thus,
necessarily $N=0$.
\end{proof}

The existence and uniqueness of the trivial solution in the time symmetric setting implies that these are natural boundary values to place on the \eqref{AKE} here. Interestingly, this coincides with the boundary values of lemma \ref{vanishingbdrylemma} so that in the time symmetric case, with the boundary values we have just prescribed, it turns out that necessarily one has a solution to the KID equations and therefore the spacetime evolving from this initial data will have a killing vector.

\subsubsection{Measuring the deviation from staticity on $\partial\mathcal{S}$}
In the previous subsection it has been shown that, in the time symmetric setting, a natural
prescription of the lapse $N$ on the stable MOTS $\partial\mathcal{S}$
is $N=0$. Moreover, the Einstein evolution equations under time symmetry imply that
$\Delta_{h}N=0$.

Given the choice
\begin{equation}
N=\Delta_{h}N=0, \qquad \mbox{on} \quad \partial\mathcal{S},
\label{BasicBoundaryConditionsTimeSymmetric}
\end{equation}
it is natural to ask how much of the (time symmetric) KID equations are
satisfied on $\partial\mathcal{S}$. Setting $N=\Delta_{h}N=0$ into
equations \eqref{nncompkid} and \eqref{nscompkid} yields 
\begin{alignat*}{2}
\mathrm{normal-normal:}&\qquad \rho_{p}\rho^{i}D_{i}(\rho_{q}\rho^{j}D_{j}N)=0,\\
\mathrm{normal-surface:}& \qquad  \rho_{p}\rho^{i}D_{i}\bD_{q}N= 0.
\end{alignat*}
For the `normal-normal' component observe that, using Gaussian
coordinates,  it follows from the conditions
\eqref{BasicBoundaryConditionsTimeSymmetric} that
\begin{align*}
0&=\rho_{p}\rho_{q}\rho^{i}\rho^{j}D_{i}D_{j}N\\
&=\rho_{p}\rho_{q}\left(h^{ij}-\bh^{ij}\right)D_{i}D_{j}N\\
&=\rho_{p}\rho_{q}\Delta_{\bh}N,
\end{align*}
so that $\Delta_{\bh}N=0$. This is consistent with equation
\eqref{kiddc2}. Thus, solving the intrinsic equation \eqref{kiddc2}
also solves the normal-normal equation. On the other hand, the
normal-surface component does not fully vanish as a consequence of the
conditions \eqref{BasicBoundaryConditionsTimeSymmetric}. Accordingly,
one obtains an extra condition that needs to be imposed to satisfy the
time symmetric KID equations on $\partial\mathcal{S}$.  The
observation is contained in the following lemma in which we derive a
\emph{Dain-like} invariant characterising staticity of of the initial
data set on the MOTS $\partial\mathcal{S}$.

\begin{lemma}
  Given time symmetric initial data set, let 
\[
N=\Delta_{h}N=0, \qquad \mbox{on} \quad \partial\mathcal{S}.
\]
Then the time symmetric KID equations are satisfied on $\partial
\mathcal{S}$ if and only if $\omega=0$ where
\begin{equation}
\omega \equiv \int_{\partial\mathcal{S}} |\bK|^{2}|\msD N|^{2},
\end{equation}
with $|\bK|^{2}=\bK_{pq}\bK^{pq}$ and $D^{\perp}\equiv \rho^{i}D_{i}$.
\end{lemma}

\begin{proof}
For the normal-surface component of the second KID equation one has
that 
\begin{align*}
0&= \bh_{p}{}^{i}\rho_{q}\rho^{j}D_{i}D_{j}N\\
&= \bh_{p}{}^{i}\rho_{q}\left(D_{i}\left(\rho^{j}D_{j}N\right)-D_{j}ND_{i}\rho^{j}\right)\\
&=\bh_{p}{}^{i}\rho_{q}\left(D_{i}\left(\rho^{j}D_{j}N\right) + D_{j}N\bK_{i}{}^{j}\right)\\
&=\rho_{q}\bD_{p}\left(\rho^{j}D_{j}N\right).
\end{align*}
Thus, this condition is equivalent to the statement that
$\rho^{i}D_{i}N$ is constant along $\partial\mathcal{S}$. We can use
this observation to construct a quantity that measures the
non-staticity of the boundary $\partial \mathcal{S}$. Taking the
$L^{2}$ norm of the quantity $\bD_{i}\msD N$ on $\Sigma$, the above
condition becomes, by integrating by parts
\[
0=\int_{\partial\mathcal{S}} \bD_{i}\msD N\bD^{i}\msD N = -\int_{\partial\mathcal{S}}\msD N \Delta_{\bh}\msD N.
\]
In order to simplify this further, recall equation \eqref{kiddc2}. Taking the normal derivative of this equation yields
\[
\msD\Delta_{\bh}N-\frac{1}{2}(\er+|\bK|^{2})\msD N=0
\]
where we have used the Leibniz rule on the second term and that $N=0$
on $\partial\mathcal{S}$. To use this expression in the
integral above, one commutes derivatives in the first term. Using the
assumption of Gaussian coordinates 
\begin{align*}
\msD\Delta_{\bh}N&=\bh^{ij}\rho^{k}D_{k}D_{i}D_{j}N\\
&=\bh^{ij}\rho^{k}(r^{l}{}_{jki}D_{l}N+D_{i}D_{k}D_{j}N).
\end{align*}
In the second term, we now use the Leibniz rule to obtain 
\begin{align*}
\bh^{ij}\rho^{k}D_{i}D_{k}D_{j}N&=\Delta_{\bh}\msD N-\bh^{ij}D_{i}(D_{j}\rho^{k} D_{k}N)-\bK^{jk}D_{k}D_{j}N\\
&=\Delta_{\bh}\msD N-\bh^{ij}D_{i}(\bK_{j}{}^{k} D_{k}N)-\bK^{pq}(\bD_{p}\bD_{q}N+\bK_{pq}\msD N)\\
&=\Delta_{\bh}\msD N-\bh^{ij}\bK_{j}{}^{k} D_{i}(\bD_{k}N) - |\bK|^{2}\msD N\\
&=\Delta_{\bh}\msD N - |\bK|^{2}\msD N
\end{align*}
where we have, again, used that $N=0$ on $\partial\mathcal{S}$ and
changed the 3-covariant derivative to the 2-covariant derivative by
contraction with an intrinsic quantity. To get to the fourth line one uses that $N=0$ on $\partial\mathcal{S}$. To take care of the
Riemann tensor, we observe that 
\begin{align*}
\bh^{ij}\rho^{k}r^{l}{}_{jki}D_{l}N&=\bh^{ij}\rho^{k}r^{l}{}_{jki}h_{l}{}^{m}D_{m}N\\
&=\bh^{ij}\rho^{k}r^{l}{}_{jki}(\bh_{l}{}^{m}+\rho_{l}\rho^{m})D_{m}N\\
&=\bh^{ij}\rho^{k}r^{l}{}_{jki}\rho_{l}\rho^{m}D_{m}N\\
&=\bh^{ij}\rho^{k}\rho^{l}r_{ljki}\msD N.
\end{align*}
Now, we can apply the Gauss-Codazzi identity to obtain
\begin{align*}
\bh^{ij}\rho^{k}\rho^{l}r_{ljki}=\frac{1}{2}(\er+|\bK|^{2}).
\end{align*}
Thus,
\[
\msD\Delta_{\bh}N=\frac{1}{2}(\er-|\bK|^{2})\msD N + \Delta_{\bh}\msD N, 
\]
so that we finally obtain
\[
 \Delta_{\bh}\msD N=\frac{1}{2}(\er+|\bK|^{2})\msD N-\frac{1}{2}(\er-|\bK|^{2})\msD N = |\bK|^{2}\msD N.
\]
Thus, we can rewrite the condition that the second KID equation is
satisfied in terms of the vanishing of 
\[
\omega \equiv \int_{\Sigma} |\bK|^{2}|\msD N|^{2}.
\]
In other words, in time symmetric initial data, if $N=\Delta_{h}N=0$
on the $\partial \mathcal{S}$ then the KID equations will be satisfied
at the boundary if and only if $\omega=0$. 
\end{proof}

\subsection{The non-time symmetric case}\label{generalcasesolvability}

Having shown in the previous section that the KID equations can be
used to choose suitable
boundary values for the approximate Killing equation as well as constructing an
invariant characterising the stationarity at the apparent horizon, we
now move on to the non-time symmetric case, $K_{ab}\neq 0$. In this
setting, the decomposition of the KID equations is much more
complicated. We can no longer consider $N^{i}=0$ and thus have to
study a system of equations. Using the time symmetric case as a
blueprint, we begin by manipulating the equations in order to obtain a
system of equations that has a similar form to the operator
$\mathcal{K}$ above so that we can investigate how much of the full
KID equations can be satisfied on a boundary $\partial\mathcal{S}$
which is assume to be a MOTS.

\subsubsection{Intrinsic equations over $\partial\mathcal{S}$}
We begin with the second decomposed KID equation
\eqref{fulldecompkideq2}. Taking the trace and using that
$\bh^{pq}\rho_{p}=0$ yields
\begin{eqnarray*}
&& N^{k}\Bigl( \bh^{pq}D_{k}K_{pq}^{\parallel}-K_{iq} \bigl(\rho^{i}
   \bK_{k}{}^{q}\bigr) - K_{pj} \bigl(\rho^{j} \bK_{k}{}^{p}
   \bigr)\Bigr)+ \bh^{ij}D_{i}N^{k}K_{kj}+\bh^{ij}D_{j}N^{k}K_{ik} +
   \Delta_{\bh}N-\rho^{k}D_{k}N\bK\\
  && \hspace{2cm} =N\left( \bh^{ij}r_{ij}+\bh^{ij}K_{ij}K+\bh^{ij}K_{ik}K\indices{^{k}_{j}}\right).
\end{eqnarray*}
One can simplify the first term as follows
\begin{align*}
\bh^{pq}D_{k}K_{pq}^{\parallel}&=D_{k}(\bh^{ij}K_{ij})-K_{ij}\bh\indices{_{p}^{i}}\bh\indices{_{q}^{j}}D_{k}(\rho^{p}\rho^{q})\\
&=D_{k}\bK,
\end{align*}
where, to get from the first to the second line, one uses
$\bh^{pq}\rho_{p}=0$ and the MOTS condition
\eqref{motscondition}. Thus, one obtains
\begin{eqnarray*}
&& N^{k}D_{k}\bK-N_{\parallel}^{A}K_{ij}\rho^{i}\bK_{A}{}^{j} -
  N^{A}K_{ij}\rho^{j} \bK_{A}{}^{i}+
   \bh^{ij}D_{i}N^{k}K_{kj}+\bh^{ij}D_{j}N^{k}K_{ik}+\Delta_{\bh}N-\rho^{k}D_{k}N\bK\\
  && \hspace{2cm} =N\left( \bh^{ij}r_{ij}+\bK K+\bh^{ij}K_{ik}K\indices{^{k}_{j}}\right).
\end{eqnarray*}

\emph{Now, we assume that the quantities $\rho_{k}N^k$ and $\pdv{}{\rho}N^{i}$ are and prescribed on $\partial\mathcal{S}$.} We separate
terms into their normal and tangential components as follows:
\begin{align*}
    N^kD_k \bK &= (\bh^k_i+\rho_i\rho^k)N^i D_k \bK\\
    &=N_{\parallel}^A\bD_A \bK +\rho_iN^i\rho^k D_k\bK,
\end{align*}
and
\begin{align*}
    h^{ij}D_i N^k K_{ij} &= h^{ij}D_i N_{\parallel}^k K_{kj} + h^{ij}D_i(\rho^k\rho_l N^l)K_kj\\
    &=\bD^A N_{\parallel}^B K^{\parallel}_{AB} - \rho^l K_{li}\bK^i_AN_{\parallel}^A+h^{ij}D_i(\rho^k\rho_l N^l)K_{kj}.
\end{align*}
Thus,  the second KID equation on $\partial\mathcal{S}$ implies that
\[
\Delta_{\bh}N - \rho^{k}D_{k}N\bK
    +2\bD^A N_{\parallel}^B K^{\parallel}_{AB} +N_{\parallel}^A(\bD_A \bK-4K_{ij}\rho^{i}\bK_{A}{}^{j}) -N\left( \bh^{ij}r_{ij}+\bK K+\bh^{ij}K_{ik}K\indices{^{k}_{j}}\right)=\tilde{F}
  \]
  where
\[
  \tilde{F}\equiv-2h^{ij}D_i(\rho^k\rho_l
  N^l)K_{kj}-\rho_iN^i\rho^k D_k\bK.
\]
In order to remove the normal derivate of $N$ in the above equation, consider the trace of the first decomposed KID equation \eqref{fulldecompkideq1}
\[
N\bK = N_{k}\rho^{k }-\bD_{A}N_{\parallel}^{A}.
\]
Taking the normal derivative of this quantity, one obtains an expression for the normal derivative of $N$ in terms of prescribed quantities and thus the second KID equation on $\partial\mcS$ can be written 
\begin{align*}
\Delta_{\bh}N
    +(2 K_{\parallel}^{AB}-\bK^{AB})\bD_A N^{\parallel}_B &+N_{\parallel}^A(\bD_A \bK-4K_{ij}\rho^{i}\bK_{A}{}^{j}-\rho^{k}r_{kA})\\& -N\left( \bh^{ij}r_{ij}+\bK K+\bh^{ij}K_{ik}K\indices{^{k}_{j}}-\rho^{k}\bD_{k}\bK\right)=F
\end{align*}
 where $F$ now includes terms involving the normal derivative of $N^{i}$.
 
\medskip
For the first KID equation, we cannot just take a derivative tangential to
$\partial\mathcal{S}$ and take the trace as we did in the time symmetric case as the resulting equation is not elliptic. Instead, we
consider the trace free part of \eqref{fulldecompkideq1}, namely
\[
NK^{\parallel}_{AB}+\bD_{(A}N^{\parallel}_{B)}+N_{k}\rho^{k}\bK_{AB}-\frac{1}{2}\bh_{AB}(\bD_C N_{\parallel}^C+\bK(N+N_k\rho^k))=0.
\]
Taking the divergence of this equations and using the special form of
the Riemann tensor in 2-dimensions we can write
\[
\Delta_{\bh}N^{\parallel}_B +(2K^{\parallel}_{AB}-\bK\bh_{AB})\bD^A N +RN^{\parallel}_B +N(2\bD^AK^{\parallel}_{AB}-\bD_B\bK)=F_B,
\]
where
\[
  F_B \equiv \bD_B(\bK N_k\rho^k)-2\bD^A(N_k\rho^k\bK_{AB}).
  \]
Thus, we have the following system of equations on
$\partial\mathcal{S}$:
\begin{subequations}
\begin{align}
&\Delta_{\bh}N^{\parallel}_B +(2K^{\parallel}_{AB}-\bK\bh_{AB})\bD^A N +RN^{\parallel}_B +N(2\bD^AK^{\parallel}_{AB}-\bD_B\bK)=F_B\label{kidwsourceterms1}\\
\begin{split}
&\Delta_{\bh}N
    +(2 K_{\parallel}^{AB}-\bK^{AB})\bD_A N^{\parallel}_B +N_{\parallel}^A(\bD_A \bK-4K_{ij}\rho^{i}\bK_{A}{}^{j}-\rho^{k}r_{kA})\\&\hspace{13.5 em} -N\left( \bh^{ij}r_{ij}+\bK K+\bh^{ij}K_{ik}K\indices{^{k}_{j}}-\rho^{k}\bD_{k}\bK\right)=F
\label{kidwsourceterms2}
\end{split}
\end{align}
\end{subequations}
where $F$ and $F_B$ are source terms. These source terms are
completely determined in terms of the intrinsic geometry of
$\partial\mathcal{S}$, the extrinsic curvatures $K_{ij}$ and
$\bar{K}_{AB}$, the normal component of $N^i$, and the normal derivative of $N^{i}$, $\pdv{}{\rho}N^{i}$. 
\begin{remark}
\emph{As in the time symmetric case, the system \eqref{kidwsourceterms1}-\eqref{kidwsourceterms2} does not incorporate all parts of the KID equation on $\partial \mcS$. Equation \eqref{kidwsourceterms1} is a formulation of the tracefree part of \eqref{fulldecompkideq1} while \eqref{kidwsourceterms2} is the trace of \eqref{fulldecompkideq2}.}
\end{remark}
\medskip
The system \eqref{kidwsourceterms1}-\eqref{kidwsourceterms2} is manifestly
elliptic for $(N, N^{\parallel}_B)$. It can be succinctly written in
the matricial form
\[
    \Delta_{\bh} \vec{N}+T^A\bD_A \vec{N}+C\cdot\vec{N}=\vec{F}
\]
where $T_A$ and $C$ are $3\times3$ matrices and
$\vec{N}\equiv(N,N^{\parallel}_A)$. We note the formal similarity of
this equation with the time symmetric equation \eqref{bkid1}. However,
in contrast to the time symmetric case, it is not clear how to connect
the solvability of the system
\eqref{kidwsourceterms1}-\eqref{kidwsourceterms2} to, for example, the
stability of the MOTS $\partial\mathcal{S}$. This is an interesting
question which falls beyond the scope of the present article. 

\medskip
In order to provide some intuition into the consequences of the system
\eqref{kidwsourceterms1}-\eqref{kidwsourceterms2}, for the rest of
this section, we make the following assumption:

\begin{assumption}\label{assumption:soln}
  The elliptic operator associated to the system
\eqref{kidwsourceterms1}-\eqref{kidwsourceterms2} as well as its
adjoint have trivial Kernel.
  \end{assumption}

  As mentioned earlier, the above assumption ensures the existence of
  a unique solution 
$(N,N^{\parallel}_A)$ to the system
\eqref{kidwsourceterms1}-\eqref{kidwsourceterms2}. Observe that this
is independently of whether the 3-manifold $\mathcal{S}$ admits a
solution to the KID equations. However, it is important to note that
this system will always be solved by a solution to the KID
equations. The values obtained as solutions to the system
\eqref{kidwsourceterms1}-\eqref{kidwsourceterms2} then provide
boundary values for the AKE boundary value problem in the following
way: $N$ and the tangential components of $N^{i}$ are obtained through
solving the system. The value of $\Delta_{h}N$ can be obtained by
using Einstein's equations and the other quantities were already
prescribed. One could obtain values for the quantities that we
prescribed here through analysing the normal-normal and
normal-tangential components of the decomposed KID equations
independently of the above analysis. These components are derived in
appendix \ref{Appendix:DecompsotionKIDs}.

\subsubsection{Constructing an invariant on $\partial\mathcal{S}$}

Constructing an invariant on a MOTS $\partial\mathcal{S}$ immersed on
a non-time symmetric hypersurface $\mathcal{S}$ is more involved than
in the time symmetric case. It is important to note that the way one
derives an invariant is not unique. We outline below one possible way
is to consider the parts of the decomposed KID equations on
$\partial\mathcal{S}$ that are not included in the system \eqref{kidwsourceterms1}-\eqref{kidwsourceterms2}.

\medskip
For example, under Assumption \ref{assumption:soln}, the solution to the system \eqref{kidwsourceterms1}-\eqref{kidwsourceterms2} may not solve the trace of the first decomposed KID equation \eqref{fulldecompkideq1}:
\[
\mathcal{Q}\equiv N\bK - N_{k}\rho^{k }+\bD_{A}N_{\parallel}^{A}=0.
\]
The non-zero value of $\mathcal{Q}$, in conjunction with the tracefree part of the second decomposed KID equation and the normal-normal and normal-tangential components of the decomposed KID equations derived in appendix \ref{Appendix:DecompsotionKIDs}, will characterise the non-stationarity of $\mcS$ on the boundary 2-surface $\partial \mcS$. In particular, the sum of the $L^{2}$ norms of these quantities will provide a geometric invariant that incorporates all of the above quantities.

\section{Conclusion}

We have shown that there exists solutions to the \eqref{AKE},
approximate Killing vectors, on asymptotically Euclidean initial data
along with boundary conditions on an inner boundary 2-dimensional
surface. We associated this boundary with a boundary of a black hole
characterised by a MOTS. In the time symmetric case,  we have then constructed invariants on this
MOTS that classify the staticity of the initial data set. In
particular, the invariant $\omega$ vanishes when there exists a
Killing vector.

\medskip
Combining the analysis of the latter sections with the main theorem
allows us to write down the following theorem in the time symmetric
case:

\begin{theorem}
Let $\lambda$ be the Dain invariant associated to the boundary value
problem
\begin{align*}
&\mathscr{P}\circ\mathscr{P}^{*}
\begin{pmatrix}
X\\
X^{i}
\end{pmatrix}
=0 \qquad \text{on } \mathcal{S},\\
&\begin{cases}
X\vert_{\partial \mathcal{S}} =0,\\
\Delta_{h} X \vert_{\partial \mathcal{S}} =0,\\
X^{i}\vert_{\partial \mathcal{S}}=0,\\
\pdv{\rho}X^{i}\vert_{\partial \mathcal{S}}=0,
\end{cases} 
 \end{align*}
in a time symmetric complete, smooth asymptotically Euclidean
initial data set for the Einstein vacuum field equations with one
 asymptotic end and an inner boundary $\partial\mathcal{S}$. Then, if
 on the one hand 
 $\lambda=0$ then the initial data is static. On the other hand, if
 the 2-surface invariant $\omega$ is non-zero on $\partial\mathcal{S}$ then the initial data cannot
 be static ---and thus $\lambda\neq 0$. 
\end{theorem}

Note that, by Lemma \ref{vanishingbdrylemma}, since we have vanishing
boundary conditions, we have a solution to the KID equations and thus
Dain's invariant vanishes.

\medskip
In the non-time symmetric case one could write down an analogue of the
above theorem with the explicit solutions coming from the system
\eqref{kidwsourceterms1}-\eqref{kidwsourceterms2} as well as
prescribed values for the other boundary values. In this case Dain's
invariant $\lambda$ would be non-zero. Further work would entail
removing Assumption 1 and the precise construction of invariants for
specific given initial data. It would also be of interest to explore
the conditions one would have to impose on the initial data in order
to guarantee existence of unique solutions to the projected KID
equations on $\partial\mathcal{S}$.

\subsection*{Acknowledgements}
We have benefited from discussions with J.L. Jaramillo on the MOTS
stability operator. Further enriching discussions with J.L. Williams
about various aspects of this project 
are acknowledged.The calculations in this article have been carried
out in the suit {\tt xAct} for the Wolfram programming language \cite{xAct}. We
have made used of routined for integration by parts in xAct originally
written by T. B\"ackdahl. We thank him for his help in adapting these
to the setting considered in this article.

\newpage
\appendix
\section{Green's Formula for the full AKE}\label{gfappendix}
We will show that the Green formula indeed gives the form found in
Lemma \ref{lemgfAKE}. Throughout this appendix we will suppress the volume and surface forms $\mathrm{d}\mu$ and $\mathrm{d}S$ as well as $\partial \mcS$ on the surface integrals for readability. Start by considering the expression
\begin{align*}
\phantom{=}& \int_{\mathcal{S}}\mathscr{P} \circ \mathscr{P}^{*}
\begin{pmatrix}
X\\
X_{i}
\end{pmatrix}
\cdot
\begin{pmatrix}
Z\\
Z^{i}
\end{pmatrix}
-
\int_{\mathcal{S}}\mathscr{P} \circ \mathscr{P}^{*}
\begin{pmatrix}
Z\\
Z^{i}
\end{pmatrix}
\cdot
\begin{pmatrix}
X\\
X_{i}
\end{pmatrix}\\
=& \int_{\mathcal{S}}\mathscr{P} \circ \mathscr{P}^{*}(X) Z+\mathscr{P} \circ \mathscr{P}^{*}(X_{i}) Z^{i}- \int_{\mathcal{S}}\mathscr{P} \circ \mathscr{P}^{*}(Z) X+\mathscr{P} \circ \mathscr{P}^{*}(Z^{i}) X_{i}.
\end{align*}
Note the slight abuse of notation in writing $\mathscr{P} \circ
\mathscr{P}^{*}(X)$ to mean the lapse component of the AKE and
$\mathscr{P} \circ \mathscr{P}^{*}(X_{i})$ to mean the shift
component. Since the AKE operator is self-adjoint, we only need to
perform integration by parts on one of the integrals and the form of
the bulk integrals will not matter. The AKE takes the form
\begin{equation*}
\hspace{-5mm}\mathscr{P}\circ\mathscr{P}^*\left(\begin{array}{c}
X\\
X_i
\end{array}\right) \equiv \left(\begin{array}{c}
2\Delta_\bmh\Delta_\bmh X-r^{ij}D_iD_jX+2r\Delta_\bmh X+\tfrac{3}{2}D^i rD_iX+(\tfrac{1}{2}\Delta_\bmh r+r_{ij}r^{ij})X\\
+D^iD^j H_{ij}-\Delta_\bmh H_k{}^k-r^{ij}H_{ij}+\bar{H} \\[1em]
D^j\Delta_\bmh D_{(i}X_{j)}+D_i\Delta_\bmh D^kX_k+D^j\Delta_\bmh F_{ij}-D_i\Delta_\bmh F_k{}^k-\bar{F}_i
\end{array}\right)
\end{equation*}
where
\begin{align*}
& \bar{H}\equiv 2(K\bar{Q}-K^{ij}\bar{Q}_{ij})+2(K^{ki}K^j{}_k-KK^{ij})\bar{\gamma}_{ij},\\
& \bar{F}_i\equiv \left(D_i K^{kj}-D^k K^j{}_i\right)\bar{\gamma}_{jk}-\left(K^k{}_i D^j-\tfrac{1}{2}K^{kj}D_i\right)\bar{\gamma}_{jk}+\tfrac{1}{2}K^k{}_i D_k \bar{\gamma}\\
&\bar{\gamma}_{ij}\equiv D_iD_jX-X r_{ij}-\Delta_\bmh X h_{ij}+H_{ij}\\
&\bar{Q}_{ij}\equiv -\Delta_\bmh(D_{(i}X_{j)}-D^kX_k h_{ij}+F_{ij})\\
&H_{ij}\equiv 2X(K^k{}_iK_{jk}-KK_{ij})-K_{k(i}D_{j)}X^k+\tfrac{1}{2}K_{ij}D_kX^k
+\tfrac{1}{2}K_{kl}D^kX^l h_{ij}-\tfrac{1}{2}X^kD_k K_{ij}+\tfrac{1}{2}X^k D_k K h_{ij}\\
& F_{ij}\equiv 2X(Kh_{ij}-K_{ij}).
\end{align*}
We consider each component separately, beginning with the \emph{lapse component}.

\subsection{The lapse component}
The lapse component of the AKE is given by 
\begin{align*}
\mathscr{P}\circ\mathscr{P}^*(X) =&\hspace{0.5em}  2\Delta_\bmh\Delta_\bmh X-r^{ij}D_iD_jX+2r\Delta_\bmh X+\tfrac{3}{2}D^i rD_iX+(\tfrac{1}{2}\Delta_\bmh r+r_{ij}r^{ij})X\\&
+D^iD^j H_{ij}-\Delta_\bmh H_k{}^k-r^{ij}H_{ij}+\bar{H}.
\end{align*}
Therefore, we perform integration by parts term by term on the expression
\[
 \int_{\mathcal{S}}\mathscr{P} \circ \mathscr{P}^{*}(X) Z.
\]

\medskip
Since we know that the operator is self-adjoint, we know the bulk
integral obtained through integration by parts. accordingly, we
concentrate our attention on the boundary terms of the associated bulk
integral in each term of $\mathscr{P}\circ\mathscr{P}^{*}$.

\begin{align*}
\int_{\mathcal{S}} 2\Delta_\bmh\Delta_\bmh X :&\oint \pdv{}{\rho}X\Delta_h Z -\oint \pdv{}{\rho}\left(\Delta_h Z\right) X +\oint \pdv{}{\rho}\left(\Delta_h X\right) Z -\oint \pdv{}{\rho}Z\Delta_h X,\\
\\
\int_{\mathcal{S}} r^{ij}D_{i}D_{j}X Z :& \oint r^{ij}\rho_iD_j\left(X\right) Z -\oint X\rho_jD_i\left(r^{ij}Z\right),\\
\\
\int_{\mathcal{S}} r\Delta_{h}X Z :& \oint rZ\pdv{}{\rho}X -\oint X\pdv{}{\rho}\left(rZ\right),\\
\\
\int_{\mathcal{S}} \frac{3}{2}D^{i}rD_{i}XZ :& \frac{3}{2}\oint \pdv{}{\rho}\left(rZ\right)X, \\
\\
\int_{\mathcal{S}}D^{i}D^{j}H_{ij}Z :& \oint 2\rho^{i}D^{j}(XB_{ij})Z - \oint2\rho^{j}XB_{ij}D^{i}Z\\
&-\frac{1}{2}\left(\oint \rho^{i}D^{j}(K_{ki}D_{j}X^{k})Z - \oint \rho^{j}K_{ki}D_{j}X^{k}D^{i}Z +\oint K_{ki}\rho_{j}X^{k}D^{j}D^{i}Z \right)\\
&-\frac{1}{2}\left(\oint \rho^{i}D^{j}(K_{kj}D_{i}X^{k})Z - \oint \rho^{j}K_{kj}D_{i}X^{k}D^{i}Z +\oint K_{kj}\rho_{i}X^{k}D^{j}D^{i}Z \right)\\
&+\frac{1}{2}\left(\oint \rho^{i}D^{j}(K_{ij}D_{k}X^{k})Z -\oint \rho^{j}K_{ij}D_{k}X^{k}D^{i}Z+\oint \rho_{k}X^{k}K_{ij}D^{j}D^{i}Z\right)\\
&+\frac{1}{2}\left(\oint \pdv{}{\rho}\left(K_{kl}D^{k}X^{l}\right)Z -\oint K_{kl}D^{k}X^{l}\pdv{}{\rho}Z + \oint \rho^{k}K_{kl}X^{l}\Delta_{h}Z \right)\\
&-\frac{1}{2}\left(\oint \rho^{i}D^{j}X^{k}D_{k}(K_{ij})Z -\oint X^{k}\rho^{j}D_{k}(K_{ij})D^{i}Z \right)\\
&+\frac{1}{2}\left(\oint \pdv{}{\rho}\left(X^{k}D_{k}K\right)Z-\oint X^{k}D_{k}K\pdv{}{\rho}Z \right),\\
\\
\int_{\mathcal{S}} \Delta_{h}H_{k}{}^{k}Z : & \oint \pdv{}{\rho}\left(XB_{k}{}^{k}\right)Z - \oint2 XB_{k}{}^{k}\pdv{}{\rho}Z \\
&+\frac{1}{2}\left(\oint  \rho_{j}D^{j}(K_{kl}D^{k}X^{l})Z - \oint \rho^{j}K_{kl}D^{k}X^{l}D_{j}Z +\oint K_{kl}\rho^{k}X^{l}D^{j}D_{j}Z  \right)\\
&+\frac{1}{2}\left(\oint \rho_{j}D^{j}(K D_{k}X^{k})Z -\oint \rho^{j}K D_{k}X^{k}D_{j}Z+\oint \rho_{k}X^{k}D^{j}D_{j}Z\right)\\
&+\oint \pdv{}{\rho}\left(X^{k}D_{k}K\right)Z-\oint X^{k}D_{k}K\pdv{}{\rho}Z, \\
\\
\int_{\mathcal{S}} r^{ij}H_{ij}Z : & \frac{1}{2}\left(-\oint r^{ij}K_{ki}\rho_{j}X^{k}Z+\oint r^{ij}K_{ij}\rho_{k}X^{k}Z+\oint rK_{kl}\rho^{k}X^{l}Z\right),\\
\\
\int_{\mathcal{S}}\bar{H}Z :&  4\left(\oint K\rho_{i}D^{i}D_{k}X^{k}Z -\oint \rho^{i}D_{k}X^{k}D_{i}(KZ) +\oint \rho_{k}X^{k}\Delta_{h}(KZ) \right)\\
&-8\left(\oint \pdv{}{\rho}X Z -\oint X\pdv{}{\rho}Z\right)\\
&+2\left( \oint \rho_{k}K^{ij}D^{k}D_{i}X_{j}Z -\oint \rho^{k}D_{i}X_{j}D_{k}(K^{ij}Z)+\oint \rho_{i}X_{j}\Delta_{h}(K^{ij}Z)\right)\\
&-2\left(\oint K\rho_{i}D^{i}D_{k}X^{k}Z -\oint n^{i}D_{k}X^{k}D_{i}(KZ) +\oint \rho_{k}X^{k}\Delta_{h}(KZ) \right)\\
&+4\left(\oint K^{ij}\rho_{k}D^{k}(A_{ij}X)Z-\oint n^{k}A_{ij}XD_{k}(K^{ij}Z)\right)\\
&+2\left(\oint B^{ij}\rho_iD_j\left(X\right) Z -\oint X\rho_jD_i\left(B^{ij}Z\right)\right)\\
&-2\left(\oint BZ\pdv{}{\rho}X -\oint X\pdv{}{\rho}\left(BZ\right)\right)\\
&+\left(-\oint r^{ij}K_{ki}\rho_{j}X^{k}Z+\oint r^{ij}K_{ij}\rho_{k}X^{k}Z+\oint rK_{kl}n^{k}X^{l}Z\right),
\end{align*}
where $A_{ij}\equiv Kh_{ij}-K_{ij}$ and $B_{ij}\equiv
K_{ik}K^{k}{}_{j}-KK_{ij}$. By inspection one sees that $X=0,
\Delta_{h}X=0, X_{i}=0, \pdv{}{\rho}X^{k}=0$ and the same for $Z$ is
enough to for all the boundary terms to vanish.

\subsection{The shift component}
The shift component of the AKE is
\[
\mathscr{P}\circ\mathscr{P}^{*}(X_{i}) =D^j\Delta_\bmh D_{(i}X_{j)}+D_i\Delta_\bmh D^kX_k+D^j\Delta_\bmh F_{ij}-D_i\Delta_\bmh F_k{}^k-\bar{F}_i.
\]
In analogy to the lapse component, we want to perform integration by parts on 
\[
 \int_{\mathcal{S}}\mathscr{P} \circ \mathscr{P}^{*}(X_{i}) Z^{i}.
\]
As above, we compute term by term ignoring the final bulk term. 
The first two terms (i.e. the ones with highest order derivatives)
have the following boundary terms:
\begin{align*}
\oint \frac{1}{2}n^{j}Z^{i}\Delta_{h}D_{i}X_{j} &- \oint \frac{1}{2}\rho_{i}\Delta_{h}D^{j}Z^{i}X_{j}\\+\oint \frac{1}{2}n^{j}Z^{i}\Delta_{h}D_{j}X_{i} &- \oint \frac{1}{2}\rho_{j}\Delta_{h}D^{j}Z^{i}X_{i}\\
+\oint \rho_{i}Z^{i}\Delta_{h}D^{k}X_{k} &- \oint n^{k}\Delta_{h}D_{i}Z^{i}X_{k}\\
+\oint \frac{1}{2}\pdv{}{\rho}\left(D^{j}Z^{i}\right)D_{j}X_{i} &-\oint \frac{1}{2}D^{j}Z^{i}\pdv{}{\rho}\left(D_{i}X_{j}\right)\\
+\oint \frac{1}{2}\pdv{}{\rho}\left(D^{j}Z^{i}\right)D_{i}X_{j} &-\oint \frac{1}{2}D^{j}Z^{i}\pdv{}{\rho}\left(D_{j}X_{i}\right)\\
+\oint \pdv{}{\rho} \left(D_{i}Z^{i}\right)D^{k}X_{k} &- \oint D_{i}Z^{i}\pdv{\rho}\left(D^{j}X_{j}\right),
\end{align*}
\begin{align*}
\int_{\mathcal{S}}D^j\Delta_\bmh F_{ij}Z^{i} : \oint n^{j}\Delta_{h}(2A_{ij}X)Z^{i} - \oint \rho_{k}D^{k}(2A_{ij}X)D^{j}Z^{i}+\oint 2A_{ij}Xn^{k}D_{k}D^{j}Z^{i},
\end{align*}

\begin{align*}
\int_{\mathcal{S}}D_i\Delta_\bmh F_k{}^k : \oint \rho_{i}\Delta_{h}(4KX)Z^{i}-\oint \rho_{k}D^{k}(4KX)D_{i}Z^{i} +\oint 4n^{k}KXD_{k}D_{i}Z^{i}.
\end{align*}
Now, let $C_{i}{}^{jk}\equiv D_{i}K^{kj}-D^{k}K_{i}{}^{j}$ so that $\bar{F}_{i}$ can be written as 
\begin{equation*}
\bar{F}_i = C_{i}{}^{jk}\bar{\gamma}_{jk}-\left(K^k{}_i D^j-\tfrac{1}{2}K^{kj}D_i\right)\bar{\gamma}_{jk}+\tfrac{1}{2}K^k{}_i D_k \bar{\gamma}.
\end{equation*}
Then
\begin{align*}
\int_{\mathcal{S}}\bar{F}_i  Z^{i} : &\left(\oint C_{i}{}^{jk}\rho_jD_k\left(X\right) Z^{i} -\oint X\rho_kD_j\left(C_{i}{}^{jk}Z^{i}\right)\right)\\
&-\left(\oint C_{i}Z^{i}\pdv{}{\rho}X -\oint X\pdv{}{\rho}\left(C_{i}Z^{i}\right)\right)\\
&+\frac{1}{2}\left(-\oint r^{ij}K_{ki}\rho_{j}X^{k}Z+\oint r^{ij}K_{ij}\rho_{k}X^{k}Z+\oint rK_{kl}n^{k}X^{l}Z\right)\\
&\\
&-\left[\oint K^{k}{}_{i}\rho_{l}D^{l}D_{k}XZ^{i} -\oint n^{l}D_{k}XD_{l}(K^{k}{}_{i}Z^{i})+\oint \rho_{k}X\Delta_{h}(K^{k}{}_{i}Z^{i})\right.\\
&-\oint n^{j}K_{i}{}^{k}Xr_{jk}Z^{i}\\
&-\left(\oint K_{i}{}^{k}\rho_{k}\Delta_{h}X Z^{i}-\oint \rho_{l}D^{l}XD_{k}(K_{i}{}^{k}Z^{i}) + \oint n^{l}XD_{l}D_{k}(K_{i}{}^{k}Z^{i})\right)\\
&+\oint 2K_{i}{}^{k}n^{j}B_{jk}XZ^{i}\\
&-\frac{1}{2}\left(\oint K_{i}{}^{k}n^{j}K_{lj}D_{k}X^{l}Z^{i} -\oint K_{lj} \rho_{k}X^{l}D^{j}(K_{i}{}^{k}Z^{i})
\right)\\
&-\frac{1}{2}\left(\oint K_{i}{}^{k}n^{j}K_{lk}D_{j}X^{l}Z^{i} -\oint K_{lk} \rho_{j}X^{l}D^{j}(K_{i}{}^{k}Z^{i})\right)\\
&+\frac{1}{2}\left(\oint K_{i}{}^{k}n^{j}K_{jk}D_{l}X^{l}Z^{i}-\oint K_{jk}\rho_{l}X^{l}D^{j}(K_{i}{}^{k}Z^{i})\right)\\
&+\frac{1}{2}\left(\oint K_{i}{}^{k}\rho_{k}K_{lm}D^{l}X^{m}Z^{i}-\oint K_{lm}n^{l}X^{m}D_{k}(K_{i}{}^{k}Z^{i})\right)\\
&-\frac{1}{2}\oint K_{i}{}^{k}n^{j}X^{l}D_{i}K_{jk}Z^{i}\\
&\left.+\frac{1}{2}\oint K_{i}{}^{k}\rho_{k}X^{l}D_{l}KZ^{i}\right]\\
&\\
&+\frac{1}{2}\left[\oint K^{kj}\rho_{i}D_{j}D_{k}XZ^{i} -\oint \rho_{j}D_{k}XD_{i}(K^{kj}Z^{i})+\oint \rho_{k}XD_{j}D_{i}(K^{kj}Z^{i})\right.\\
&-\oint \rho_{i}K^{kj}Xr_{jk}Z^{i}\\
&-\left(\oint K\rho_{i}\Delta_{h}X Z^{i}-\oint \rho_{l}D^{l}XD_{i}(KZ^{i}) + \oint n^{l}XD_{l}D_{i}(KZ^{i})\right)\\
&+\oint 2K^{kj}\rho_{i}B_{jk}XZ^{i}\\
&-\frac{1}{2}\left(\oint K^{kj}\rho_{i}K_{lj}D_{k}X^{l}Z^{i} -\oint K_{lj} \rho_{k}X^{l}D_{i}(K^{kj}Z^{i})
\right)\\
&-\frac{1}{2}\left(\oint K^{kj}\rho_{i}K_{lk}D_{j}X^{l}Z^{i} -\oint K_{lk} \rho_{j}X^{l}D_{i}(K^{kj}Z^{i})\right)\\
&+\frac{1}{2}\left(\oint K^{kj}\rho_{i}K_{jk}D_{l}X^{l}Z^{i}-\oint K_{jk}\rho_{l}X^{l}D_{i}(K^{kj}Z^{i})\right)\\
&+\frac{1}{2}\left(\oint K \rho_{i}K_{lm}D^{l}X^{m}Z^{i}-\oint K_{lm}n^{l}X^{m}D_{i}(KZ^{i})\right)\\
&-\frac{1}{2}\oint K^{kj}\rho_{i}X^{l}D_{l}K_{jk}Z^{i}\\
&\left.+\frac{1}{2}\oint K\rho_{i}X^{l}D_{l}KZ^{i}\right]\\
&\\
&+\frac{1}{2}\left[-\left(\oint K_{i}{}^{k}\rho_{k}\Delta_{h}XZ^{i} -\oint \rho_{l}D^{l}XD_{k}(K_{i}{}^{k}Z^{i}) +\oint n^{l}XD_{l}D_{k}(K_{i}{}^{k}Z^{i})\right) \right.\\
&-\frac{1}{2}\oint K_{i}{}^{k}\rho_{k}XrZ^{i}\\
&+\oint K_{i}{}^{k}\rho_{k}XBZ^{i}\\
&+\frac{1}{4}\left(\oint K_{i}{}^{k}\rho_{k}K_{lm}D^{l}X^{m}Z^{i}-\oint n^{l}X^{m}K_{lm}D_{k}(K_{i}{}^{k}Z^{i})\right)\\
&+\frac{1}{4}\left(\oint K_{i}{}^{k}\rho_{k}KD_{l}X^{l}Z^{i}-\oint \rho_{l}X^{l}KD_{k}(K_{i}{}^{k}Z^{i})\right)\\
&\left.+\frac{1}{2}\oint K_{i}{}^{k}\rho_{k}X^{l}D_{l}KZ^{i}\right].
\end{align*}
We see that setting $X=0$, $\Delta_h X=0$, $X^i=0$ and
$\pdv{}{\rho}X^i=0$ and, equivalently, for $Z$: $Z=0$, $\Delta_h Z=0$,
$Z^i=0$ and $\pdv{}{\rho}Z^i=0$ that all boundary integrals in the above expressions
vanish.  This assertion can be verified by using the fact that we can separate the derivative
terms like $D^{i}X^{j}$ into an tangential part and normal part to the
boundary $\partial \mcS$. Integrating by parts on the intrinsic
derivatives yields $X^{i}$ and the normal derivative part. Both of
these vanish using the vanishing of $X^i$ and its normal
derivative.

\medskip
In summary, we can write Green's formula as
\begin{align*}
\phantom{=}& \int_{\mathcal{S}}\mathscr{P} \circ \mathscr{P}^{*}
\begin{pmatrix}
X\\
X_{i}
\end{pmatrix}
\cdot
\begin{pmatrix}
Z\\
Z^{i}
\end{pmatrix}
-
\int_{\mathcal{S}}\mathscr{P} \circ \mathscr{P}^{*}
\begin{pmatrix}
Z\\
Z^{i}
\end{pmatrix}
\cdot
\begin{pmatrix}
X\\
X_{i}
\end{pmatrix}\\
=&\sum_{j=1}^2\sum_{\alpha=1}^{4}\left( \oint S^{\alpha}_{j}(X,X^i)B'^{\alpha}_{j}(Z,Z^i)-\oint b^{\alpha}_j(X,X^i)T^{\alpha}_{j}(Z,Z^i)\right)
\end{align*}
where
\begin{align*}
&b_1^{1} = X\\
&b_2^{1} = \Delta_h X\\
&b_{1}^{2,3,4}=X^{i}\\
&b_{2}^{2,3,4}=\pdv{}{\rho}X^{i}
\end{align*}
and
\begin{align*}
&B'^{1}_1 = Z\\
&B_2^{'1} = \Delta_h Z\\
&B_{1}^{'2,3,4}=Z^i\\
&B_{2}^{'2,3,4}=\pdv{}{\rho}Z^i
\end{align*}
Thus, we verify that the boundary operators satisfy $b_j^{\alpha}=B_j^{;\alpha}$ and,
accordingly, the associated boundary value problems is self-adjoint. 

\section{Deriving the solution to the ODE arising from the LS condition}\label{LSODEappendix}
 In this appendix we derive the solution to the system of equations 
 \[
\begin{cases}
\left(\displaystyle\dv[2]{}{\rho}-|\xi|^{2}\right)\left(4 \displaystyle\dv[2]{}{\rho}X -|\xi|^{2}X+3\mi\xi^{A}\displaystyle\dv{}{\rho}X_{A}\right)=0,\\
\left(\displaystyle\dv[2]{}{\rho}-|\xi|^{2}\right)\left(\left(\displaystyle\dv[2]{}{\rho}-|\xi|^{2}\right)X_{A}+3\mi\xi_{A}\left(\displaystyle\dv{}{\rho}X+\mi\xi^{B}X_{B}\right)\right)=0.
\end{cases}
\]
This system is used in the analysis of the Lopatinskij-Shapiro
conditions in Section \ref{Subsection:LS}.

\medskip
Using the Ansatz
\[
X_{i}=\sum_{n=0}^{k}
X_{*i}\rho^{n}\me^{\pm |\xi| \rho}
\]
where $k=0,1,2$, we obtain the general
solution in vector notation
\begin{align}\label{LP2sol}
\begin{split}
\vec{X}=&\vec{a}\me^{ |\xi| \rho}+\vec{b}\me^{ -|\xi| \rho}+\alpha
\begin{pmatrix}
c\\
c_{1}\\
\frac{-c_{1}\xi_{1}+\mi c |\xi|}{\xi_{2}}
\end{pmatrix}
\rho\me^{ |\xi| \rho}
+\beta
\begin{pmatrix}
d\\
d_{1}\\
\frac{-d_{1}\xi_{1}-\mi d |\xi|}{\xi_{2}}
\end{pmatrix}
\rho\me^{ -|\xi| \rho}\\
&
+\gamma
\left( \rho^{2}
\begin{pmatrix}
-\frac{3r|\xi|}{10}-\frac{3\mi r_{1}(\xi_{1}+\xi_{2})}{10}\\
\frac{3}{10}\xi_{1}\left(-\mi r+\frac{r_{1}(\xi_{1}+\xi_{2})}{|\xi|}\right)\\
\frac{3}{10}\xi_{2}\left(-\mi r+\frac{r_{1}(\xi_{1}+\xi_{2})}{|\xi|}\right)
\end{pmatrix}
+
\rho
\begin{pmatrix}
r\\
r_{1}\\
r_{2}
\end{pmatrix}
\right)\me^{ |\xi| \rho}\\
&
+\delta
\left( \rho^{2}
\begin{pmatrix}
\frac{3s|\xi|}{10}-\frac{3\mi s_{1}(\xi_{1}+\xi_{2})}{10}\\
\frac{3}{10}\xi_{1}\left(-\mi s-\frac{s_{1}(\xi_{1}+\xi_{2})}{|\xi|}\right)\\
\frac{3}{10}\xi_{2}\left(-\mi s+\frac{s_{1}(\xi_{1}+\xi_{2})}{|\xi|}\right)
\end{pmatrix}
+
\rho
\begin{pmatrix}
s\\
s_{1}\\
s_{2}
\end{pmatrix}
\right)\me^{ -|\xi| \rho},
\end{split}
\end{align}
where $\vec{a}$ and $\vec{b}$ are constant vectors,
$\{c,c_{1},d,d_{1},r,r_{1},r_{2},s,s_{1},s_{2}\}$ are constants and
$\alpha$, $\beta$, $\gamma$, $\delta$ are constants of multiplicity. We write the vector $\vec{a}$ as 
\begin{equation*}
\vec{a}=
\begin{pmatrix}
a\\
a_{1}\\
a_{2}
\end{pmatrix}
=a
\begin{pmatrix}
1\\
0\\
0
\end{pmatrix}
+
a_{1}
\begin{pmatrix}
0\\
1\\
0
\end{pmatrix}
+a_{2}
\begin{pmatrix}
0\\
0\\
1
\end{pmatrix}
\end{equation*}
where these vectors are linearly independent. This is the same as
setting the values of $a,a_{1},a_{2}$ and using the multiplicity of
this solution of the system of ODEs to then multiply each vector by a
constant. We can do the same thing with the vector
$\vec{c}=(c,c_{1},c_{2})$ where $c_{2}=\frac{-c_{1}\xi_{1}+\mi c
|\xi|}{\xi_{2}}$. Since there are two constants and $c_{2}$
is a linear combination of $c$ and $c_{1}$ there are two linearly
independent vectors we can construct from this. Rearranging the
expression for $c_{2}$ so that $c_{1}$ and $c_{2}$ are free, one can
then make the choice of $c_{1}=1,c_{2}=0$ and $c_{1}=0,c_{2}=1$ to get
two solutions
\begin{equation}\label{LPpsol}
\gamma
\begin{pmatrix}
\frac{-\mi\xi_{1}}{|\xi|}\\
1\\
0
\end{pmatrix}\rho\me^{ |\xi| \rho}
+
\delta
\begin{pmatrix}
\frac{-\mi\xi_{2}}{|\xi|}\\
0\\
1
\end{pmatrix}\rho\me^{ |\xi| \rho}.
\end{equation}
Finally, one notes that in the final two terms of the full solution
above, equation \eqref{LP2sol}, that the vector
$\vec{r}=(r,r_{1},r_{2})$ (and $\vec{s}=(s,s_{1},s_{2})$) has to be
linearly independent to the two vectors in \eqref{LPpsol}. Choosing
$r=1,r_{1}=r_{2}=0$, the term becomes
\[
\gamma\left(\rho^{2}
\begin{pmatrix}
-\frac{3|\xi|}{10}\\
-\frac{3}{10}\mi\xi_{1}\\
-\frac{3}{10}\mi\xi_{2}
\end{pmatrix}
+
\rho
\begin{pmatrix}
1\\0\\0
\end{pmatrix}
\right)
\me^{ |\xi| \rho}.
\]
One can then perform these manipulations to all the terms with negative exponential solution to obtain the full general solution the ODE system
\begin{align*}
\vec{X}(\rho) =& 
c_{1}
\begin{pmatrix}
1\\
0\\
0
\end{pmatrix}
\me^{ -|\xi| \rho}
+
c_{2}
\begin{pmatrix}
0\\
1\\
0
\end{pmatrix}\me^{ -|\xi| \rho}
+c_{3}
\begin{pmatrix}
0\\
0\\
1
\end{pmatrix}\me^{ -|\xi| \rho}
+c_{4}
\begin{pmatrix}
\frac{\mi\xi_{1}}{|\xi|}\\
1\\
0
\end{pmatrix}\rho\me^{-|\xi| \rho}
+
c_{5}
\begin{pmatrix}
\frac{\mi\xi_{2}}{|\xi|}\\
0\\
1
\end{pmatrix}\rho\me^{-|\xi| \rho}\\
&
c_{6}\left(\frac{3}{10}
\begin{pmatrix}
|\xi|\\
-\mi\xi_{1}\\
-\mi\xi_{2}
\end{pmatrix}\rho^{2}
+
\begin{pmatrix}
1\\0\\0
\end{pmatrix}p
\right)
\me^{-|\xi| \rho}
+c_{7}
\begin{pmatrix}
1\\
0\\
0
\end{pmatrix}\me^{|\xi| \rho}
+
c_{8}
\begin{pmatrix}
0\\
1\\
0
\end{pmatrix}\me^{|\xi| \rho}
+c_{9}
\begin{pmatrix}
0\\
0\\
1
\end{pmatrix}\me^{|\xi| \rho} \\
&
c_{10}
\begin{pmatrix}
\frac{-\mi\xi_{1}}{|\xi|}\\
1\\
0
\end{pmatrix}\rho\me^{ |\xi| \rho}
+
c_{11}
\begin{pmatrix}
\frac{-\mi\xi_{2}}{|\xi|}\\
0\\
1
\end{pmatrix}\rho\me^{ |\xi| \rho}
+
c_{12}\left(-\frac{3}{10}
\begin{pmatrix}
|\xi|\\
\mi\xi_{1}\\
\mi\xi_{2}
\end{pmatrix}\rho^{2}
+
\begin{pmatrix}
1\\0\\0
\end{pmatrix}\rho
\right)
\me^{ |\xi| \rho}.
\end{align*}
with 12 constants. We can then write the stable solution as the first six terms of the this solution:
\begin{align*}
\vec{X}_{s}(\rho) =& 
c_{1}
\begin{pmatrix}
1\\
0\\
0
\end{pmatrix}
\me^{ -|\xi| \rho}
+
c_{2}
\begin{pmatrix}
0\\
1\\
0
\end{pmatrix}\me^{ -|\xi| \rho}
+c_{3}
\begin{pmatrix}
0\\
0\\
1
\end{pmatrix}\me^{ -|\xi| \rho}
+c_{4}
\begin{pmatrix}
\frac{\mi\xi_{1}}{|\xi|}\\
1\\
0
\end{pmatrix}\rho\me^{-|\xi| \rho}
+
c_{5}
\begin{pmatrix}
\frac{\mi\xi_{2}}{|\xi|}\\
0\\
1
\end{pmatrix}\rho\me^{-|\xi| \rho}\\
&
c_{6}\left(\frac{3}{10}
\begin{pmatrix}
|\xi|\\
-\mi\xi_{1}\\
-\mi\xi_{2}
\end{pmatrix}\rho^{2}
+
\begin{pmatrix}
1\\0\\0
\end{pmatrix}\rho
\right)
\me^{-|\xi| \rho}.
\end{align*}
By relabelling the constants, one obtains the form of solution found
in the proof of Lemma \ref{lslemma}. 

\section{Deriving the decomposed KID equations}
\label{Appendix:DecompsotionKIDs}

One can project the KID equations~\eqref{kid1}-\eqref{kid2} onto
the $2$-surfaces of constant $\varsigma$. Let ${N^{\parallel}}^{i}$ denote the projection of
the shift vector $N^i$ onto the $2$-surfaces ---i.e. 
${N^{\parallel}}^{i}\equiv \bh\indices{^{i}_{j}}N^{j}$. We proceed now
to decompose objects into parts perpendicular and parallel to the
normal $\rho^{i}$. Observing that 
\begin{align*}
D_{i}N_{j}&=D_{i}\left(\delta\indices{_{j}^{k}}N_{k}\right)\\
&=D_{i}((\bh_{j}{}^{k}+\rho_{j}\rho^{k})N_{k}),
\end{align*}
one concludes that 
\[
D_{i}N_{j}=D_{i}N^{\parallel}_{j}+\rho_{j}\rho^{k}D_{i}N_{k}-N_{k}\rho^{k}(\bar{K}_{ij}-\rho_{i}a_{j})-N_{k}\rho_{j}(\bar{K}\indices{_{i}^{k}}-\rho_{i}a^{k}).
\]
Thus, the projection of $D_{i}N_{j}$ onto the 2-surfaces is given by
\[
\bh\indices{_{p}^{i}}\bh\indices{_{q}^{j}}D_{i}N_{j} = \bD_{p} N^{\parallel}_{q}+N_{k}\rho^{k}\bK_{pq}.
\]
Making use of this expression on the projection of the first KID
equation~\eqref{kid1} with
$\bh\indices{_{p}^{i}}\bh\indices{_{q}^{j}}$ yields
\[
NK^{\parallel}_{pq}+\bD_{(p}N^{\parallel}_{q)}+N_{k}\rho^{k}\bK_{pq}=0,
\]
where
$K^{\parallel}_{pq}\equiv\bh\indices{_{p}^{i}}\bh\indices{_{q}^{j}}K_{ij}$. 

\medskip
For the second KID equation \eqref{kid2}, we  project term by term to derive the following using Gaussian normal coordinates:
\begin{alignat*}{2}
&\bh\indices{_{p}^{i}}\bh\indices{_{q}^{j}}D_{i}D_{j}N=\bD_{p}\bD_{q}N+\rho^{k}\bK_{pq}D_{k}N,\\\\
&\bh\indices{_{p}^{i}}\bh\indices{_{q}^{j}}N^{k}D_{k}K_{ij}=N^{k}\Bigl( D_{k}K_{pq}^{\parallel}+K_{iq} \bigl(\rho^{i} (\bK_{kp} + {a}_{p} \rho_{k}) + (\bK_{k}{}^{i} +  {a}^{i} \rho_{k}) \rho_{p}\bigr)\\ 
&+ K_{pj} \bigl(\rho^{j} (\bK_{kq} +  {a}_{q} \rho_{k}) + (\bK_{k}{}^{j} +  {a}^{j} \rho_{k}) \rho_{q}\bigr)
 -K_{ij} \bigl(\rho^{j} (\bK_{k}{}^{i} + {a}^{i} \rho_{k}) \rho_{p} \rho_{q} \\&+ \rho^{i} (\bK_{k}{}^{j} + {a}^{j} \rho_{k}) \rho_{p} \rho_{q} + \rho^{i}\rho^{j} (\bK_{kq} +  {a}_{q} \rho_{k}) \rho_{p} + \rho^{i}\rho^{j} (\bK_{kp} +  {a}_{p} \rho_{k}) \rho_{q}\bigr)\Bigr),\\\\
&2\bh^{ij}r_{ij}=\er-\bK^{2}+\bK_{pr}\bK^{pr},
\end{alignat*}
where we have used the Gauss-Codazzi equation to derive the final
identity. Putting these expression together, using Gaussian normal coordinates, we
obtain the projection of the second KID equation onto $\partial\mathcal{S}$:
\begin{align*}
\begin{split}
&N^{k}\Bigl( D_{k}K_{pq}^{\parallel}+K_{iq} \bigl(\rho^{i} (\bK_{kp}) + (\bK_{k}{}^{i}) \rho_{p}\bigr) + K_{pj} \bigl(\rho^{j} (\bK_{kq}) + (\bK_{k}{}^{j} ) \rho_{q}\bigr)\\
 &- K_{ij} \bigl(\rho^{j} (\bK_{k}{}^{i}) \rho_{p} \rho_{q} + \rho^{i} (\bK_{k}{}^{j} ) \rho_{p} \rho_{q} + \rho^{i}\rho^{j} (\bK_{kq}) \rho_{p} +  \rho^{i}\rho^{j} (\bK_{kp} ) \rho_{q}\bigr)\Bigr) \\
 &+\bh\indices{_{p}^{i}}\bh\indices{_{q}^{j}}D_{i}N^{k}K_{kj}+\bh\indices{_{p}^{i}}\bh\indices{_{q}^{j}}D_{j}N^{k}K_{ik} + \bD_{p}\bD_{q}N+\rho^{k}D_{k}N\bK_{pq}\\
 &=N\left( \bh\indices{_{p}^{i}}\bh\indices{_{q}^{j}}r_{ij}+KK^{\parallel}_{pq}+\bh\indices{_{p}^{i}}\bh\indices{_{q}^{j}}K_{ik}K\indices{^{k}_{j}}\right).
 \end{split}
\end{align*}
Together, equations \eqref{fulldecompkideq1} and
\eqref{fulldecompkideq2} give conditions on $\partial\mathcal{S}$ that Killing
vectors of the the spacetime must satisfy on this 2-surface. However,
the converse is not true. If $(N,N^i)$ a solution to these two
equations that does not necessarily mean that one has a solution to
the KID equations. This is because there are two other components of
the KID equations in this decomposition which we call the
normal-normal and normal-tangential components. For the first KID equation
\eqref{kid1}, one has
\begin{alignat*}{2}
\mathrm{normal-normal:}&\qquad NK^{\perp}_{pq} &&+\rho_{(p\vert}\rho^{j}D_{j\vert}N^{\perp}_{q)}-N_{i}a^{i}\rho_{p}\rho_{q}-N^{\perp}_{(p}a_{q)}=0\\
\mathrm{normal-tangential:}& \qquad NK_{pq}^{\perp\parallel}&&+\frac{1}{2}\left(\rho_{p}\rho^{i}D_{i}N^{\parallel}_{q} +N_{j}a^{j}\rho_{p}\rho_{q}+N^{\perp}_{p}a_{q}\right)\\
& &&+\frac{1}{2}\left(\bh_{q}{}^{j}D_{j}N^{\perp}_{p}-N_{i}\rho_{p}\bK_{q}{}^{i}-\rho^{i}N_{i}\bK_{pq}\right)=0,\nonumber
\end{alignat*}
where $K_{pq}^{\perp}\equiv\rho_{p}\rho^{i}\rho_{q}\rho^{j}K_{ij}$ and
$K_{pq}^{\perp\parallel}\equiv
\rho_{p}\rho^{i}\bh_{j}{}^{q}K_{ij}$. For the second KID equation
\eqref{kid2}, one can obtain the projection under the assumption of
Gaussian coordinates. The normal-normal component of the first term of
\eqref{kid2} can be written as  
\begin{align*}
\rho_{p}\rho_{q}\rho^{i}\rho^{j}N^{k}D_{k}K_{ij}&=\rho_{p}\rho_{q}N^{k}D_{k}(\rho^{i}\rho^{j}K_{ij})-\rho_{p}\rho_{q}N^{k}K_{ij}(\rho^{i}\bK_{k}{}^{j}+\rho^{j}\bK_{k}{}^{i})\\
&=\rho_{p}\rho_{q}N^{k}D_{k}(K-\bK)-\rho_{p}\rho_{q}N^{k}K_{ij}(\rho^{i}\bK_{k}{}^{j}+\rho^{j}\bK_{k}{}^{i}),
\end{align*}
where we have used the MOTS condition in the final line. Then the normal-normal component of the second KID equations is
\begin{align*}
\begin{split}
\mathrm{normal-normal}\qquad &\rho_{p}\rho_{q}N^{k}D_{k}(K-\bK)-\rho_{p}\rho_{q}N^{k}K_{ij}(\rho^{i}\bK_{k}{}^{j}+\rho^{j}\bK_{k}{}^{i})\\
+&\rho_{p}\rho_{q}\rho^{i}\rho^{j}(D_{i}N^{k}K_{kj}+D_{j}N^{k}K_{ik})+ \rho_{p}\rho^{i}D_{i}(\rho_{q}
\rho^{j}D_{j}N)\\
=&N(-\bh_{p}{}^{i}\bh_{q}{}^{j}r_{ij}+KK^{\perp}_{pq}+\rho_{p}\rho_{q}\rho^{i}\rho^{j}K_{ik}K_{j}^{k}).
\end{split}
\end{align*}
For the normal-tangential component, we obtain
\begin{align*}
\begin{split}
\mathrm{normal-tangential} \qquad &h_{p}{}^{i}\rho_{q}\rho^{j}(N^{k}D_{k}K_{ij}+D_{i}N^{k}K_{kj}+D_{j}N^{k}K_{ik})\\
+& \rho_{p}\rho^{i}D_{i}\bD_{q}N\\
=&N(h_{p}{}^{i}\rho_{q}\rho^{j}r_{ij}+KK^{\perp\parallel}+h_{p}{}^{i}\rho_{q}\rho^{j}K_{i}{}^{k}K_{kj}).
\end{split}
\end{align*}

\end{document}